\documentclass[10pt,journal,compsoc]{IEEEtran}

\usepackage{graphicx}
\usepackage{graphics}
\usepackage{epstopdf}
\usepackage{amssymb, amsthm, amsmath}
\usepackage[linesnumbered]{algorithm2e}

\usepackage{verbatim}
\usepackage{url}
\usepackage{cases}
\usepackage{color}
\usepackage{slashbox}
\usepackage{wrapfig}

\begin{document}
\newtheorem{definition}{Definition}
\newtheorem{problem}{Problem}
\newtheorem{theorem}{Theorem}
\newtheorem{lemma}{Lemma}
\newtheorem{corollary}{Corollary}

\title{Similarity Join and Similarity Self-Join Size Estimation in a Streaming Environment} 


\author{Davood~Rafiei
 and Fan Deng
 \IEEEcompsocitemizethanks{
   \IEEEcompsocthanksitem D. Rafiei is with the Department of Computing Science, University of Alberta, Edmonton, AB, Canada T6G 2E1. E-mail: drafiei@ualberta.ca
    \IEEEcompsocthanksitem F. Deng was with the Department of Computing Science, University of Alberta, Edmonton, AB, Canada T6G 2E1. E-mail: dengfan@hotmail.com}
    }



\IEEEtitleabstractindextext{
\begin{abstract}
We study the problem of similarity self-join and similarity join size estimation in a streaming setting where the goal is to estimate, in one scan of the input and with sublinear space in the input size, the number of record pairs that have a similarity within a given threshold. The problem has many applications in data cleaning and query plan generation, where the cost of a similarity join may be estimated before actually doing the join. 
On unary input where two records either match or don't match, the problem becomes join and self-join size estimation for which one-pass algorithms are readily available.
Our work addresses the problem for $d$-ary input, for $d \geq 1$, where the degree of similarity can vary from 1 to $d$. We show that our proposed algorithm gives an accurate estimate and scales well with the input size. We provide error bounds and time and space costs, and conduct an extensive experimental evaluation of our algorithm, comparing its estimation accuracy to a few competitors, including some multi-pass algorithms. Our results show that given the same space, the proposed algorithm has an order of magnitude less error for a large range of similarity thresholds.
\end{abstract}
 }




\maketitle
\newcommand*{\Archive}.  

\IEEEdisplaynontitleabstractindextext
\IEEEpeerreviewmaketitle
\section{Introduction}
The problem of {\it near-duplicates} or {\it similar record pairs} is associated with having multiple 
representations or records of the same entity or concept in a database. 
Detecting near-duplicates has been studied in the past
under different names 
such as data deduplication, merge-purge, record linkage, etc. \cite{ElmagarmidIV07tkde}.

Analyzing the similarity self-join size provides important insight when the semantics of rows and columns is less-known, and this is a commonly expected case in open data~\cite{RMiller2018}.
Consider, for example, an open data of bibliographic records with untagged attributes {\em title, author, journal, volume} and {\em year}.
The similarity self-join size
with a match expected in 4 out of 5 columns (i.e. 80\% similarity) gives the degree of uniformity under any  projection of 4 attributes. It can be noted that
the field {\em title} is not expected to have many duplicates, whereas the {\em author} field may have a limited number of duplicates since two authors can have the same name or an author can have more than one record.  More duplicates are expected in columns {\em journal, volume} and {\em year}. 
For the same reason, the similarity self-join size under projections of 4 attributes is not expected to be much different than that of 5, but the similarity self-join size for projections of size 3 is expected to be much larger; the same is observed for real DBLP records in our experiments (see Table~\ref{table:dataStat}). Such information about columns and their relationships can be obtained by analyzing the sizes of similarity self-joins.
In other words, the similarity self-join size describes not only the frequency distribution of the rows but also the soft dependencies and functional relationships between the columns~\cite{chaudhuri2009systems}.  
 
Many other applications can be listed for similarity join size estimation. 
The self-join size gives the degree of uniformity or skew, and the similarity self-join size gives the degree of skew under some {\em projections}.
The degree of skew is an important statistics in parallel and distributed database applications and may determine the choice of data partitioning strategies (e.g. vertical or horizontal) and algorithms being employed~\cite{DNSS92,SGKNM08}. For example, the Hadoop-based algorithm that won the Terasort benchmark in 2008 included a partitioning function that heavily relied on the
key distributions present in the 2008 benchmark, which may not be present
in other datasets~\cite{KRBH13}. A more general solution is expected to detect the skew,
which often arise in the `reduce' phase, and to balance the load accordingly.
In projected clustering~\cite{AggarwalHWY05}, one needs to find the set of dimensions such that the spread (or dissimilarity) is the least, and a similarity join size can be an important statistics in detecting those dimensions.
In general, 
estimating the skew
can also help with data storage and indexing~\cite{zfs-dedup}, data cleaning~\cite{shawnJeffery08} and maybe homogeneity analysis~\cite{giniInd89}. 
For example, before running a data deduplication
that can take days or even weeks, one may want to quickly find 
out if there are enough near-duplicates and that running an expensive cleaning
operation is justified.  

The setting we assume in this paper is that the synopsis of a table must be constructed in one pass.
This is desirable for rapidly growing tables where a multi-pass method can be costly.
Also the input data to a similarity join sometimes consists of data streams, which may not be fully available when a synopsis is constructed. For example, a similarity join placed in a join tree may take its input from other operators, and while the input to the join is streaming in, estimates of the join size may be sought.
It has been noted that cardinality estimation errors in a query cost model can easily be in multiple orders of magnitude and join queries are usually the largest contributors to those misestimates~\cite{leis2015good}.
Morales and Gionis~\cite{MoGi16} cite trend detection and near-duplicate filtering in a microblogging site as some applications of similarity self-join in a streaming setting. 

The problem addressed in this paper can be stated as follows:
given a collection of records and a threshold, estimate the
number of record pairs that have a Hamming distance equal to or less than the
threshold. The Hamming distance function is well-defined on bit strings and binary vectors and naturally extends to
more general strings, vectors, multi-field records, etc.
For example, the Hamming distance between two records gives the minimum 
number of substitutions that would transform one record to the other.
This quantity can be divided by the
number of fields to get the fraction of fields or features in
which two records differ. One minus that fraction will give what may be referred to as the {\it Hamming similarity}. Also other distance functions may be mapped to
the Hamming distance and our method can be used under those mappings.
 Many applications of Hamming distance are reported in the literature, for example to detect duplicate Web pages~\cite{Manku07}, duplicate records in academic 
digital libraries~\cite{Williams2013}, duplicate images~\cite{KSH2004}, etc.

A naive algorithm to estimate the number of near-duplicates is to compare every 
pair of records, which requires storing the entire dataset and has a quadratic time complexity.
However, an exact estimate often is not needed and an algorithm
that more efficiently obtains an estimate may be preferred~\cite{Broder97,CGGM03,LNS11}. 
Also the memory used for computing an estimate is usually limited, and storing the data structures in external memory has additional overhead which grows with the dataset size and should be avoided.  
To the best of our knowledge, our 
method is the first that computes an estimate of similarity self-join size within only one scan over 
such data and with a limited storage.
Also our experimental
results show that the error of our estimates is often less than or comparable
to some of the latest multi-pass algorithms. 



Our contributions can be summarized as follows:
\begin{itemize}
\item We study and address the problem of similarity self-join size estimation in a streaming setting where the input, given one by one, cannot be fully stored. While the problem has been studied for input that consists of 1-dimensional records such as numbers~\cite{AMS99JCSS,JoinSize99PODS,SketchNet05VLDB}, we are not aware of any previous work that addresses the problem on input that has more than one dimension. This paper studies the general problem and presents an efficient and elegant probabilistic algorithm, extending previous techniques 
to streaming input with more than one dimension. Extending our algorithm to similarity join size estimation is straightforward (as discussed in Section~\ref{sec:ext-join}).

\item We analyze the time and space costs of our algorithm, showing that the join size can be accurately estimated in logarithmic space as long as the input dimensionality is kept low (10 or less as shown in our analysis and experiments). More precisely, we prove that our algorithm gives an unbiased
accurate estimate with high probability using a set of counters and that the number of
those counters is independent of $n$, the number of records in the dataset,
and only depends on $d$, the dimensionality; the space needed for each counter is bounded to $O(\log(n))$ bits. 

\item We evaluate the performance of our method in terms of the accuracy of the estimates and the running time. Compared to random sampling which is the only competitor in a streaming setting, our algorithm gives significantly more accurate results and scales better for large input sizes. 
We also compare the performance of an ``offline'' version of our algorithm in which the intermediate results are materialized (and not sketched)  to two recently proposed non-streaming algorithms
\cite{LNS09,LNS11}. 
Our experiments on real data show that the proposed algorithm (under a comparable space requirement) is much more accurate than these alternative algorithms; it has to be pointed out that both these methods, unlike ours, require more than one pass over  data.
\end{itemize}




\subsection{Definitions and the problem statement}
\label{sec:defs-and-stmt}
Given a bag of records where $n$ denotes the number of distinct records and $m_i$ denotes the multiplicity of record $i$, the self-join size of the bag (also referred to as the second frequency moment) is defined as
\begin{equation}
f_2=\sum_{i=1}^n m_i^2.
\label{eq:selfjoin}
\end{equation}
Given a pair of records with the same schema, we call the pair {\it s-similar} if the number of attributes where the pair have the same values is $s$.
For examples, records $r_1$ and $r_3$ in Table~\ref{table:4rowex} are 2-similar, and so are records $r_2$ and $r_4$. 
Compared to the Hamming distance which is defined on binary vectors (i.e. vectors with 0/1 values),
 $s$-similarity is defined on general records (e.g. of employees or students).

Extending Eq.~\ref{eq:selfjoin}, consider the self-join of a collection of $d$-dimensional records, and let $x_k$ denote the number of different record pairs that are $k$-similar. As in self-join, $k$-similar pairs ($r_1$,$r_2$) contribute twice to $x_k$, but self-pairs ($r_1$,$r_1$) are not counted. Hence, the number of self-pairs is added to $s$-similarity self-join size, $g_s$, defined as 
\begin{equation}
g_s = \sum_{k=s}^d x_k + n,
\label{eq:gsDef}
\end{equation}
for $n$ records, each of dimensionality $d$. $g_s$ gives the number of record pairs that are at least $s$-similar. 
In a streaming setting, similarity join may be computed at any point while the stream is being received (aka continuous queries). 



{\bf Problem statement}.
Given $n$ records each with $d$ attributes and a similarity threshold $s$, we want to estimate $g_s$,
the $s$-similarity self-join size in one scan of the input with a limited memory smaller than $n$. We also want to extend our estimation to similarity join sizes. 



{\bf Organization}. The outline of the rest of the paper is as follows:
Three baseline algorithms are briefly presented in Section~\ref{sec:baseline-algs},
and our similarity self-join size estimation algorithm is discussed in
Section~\ref{sec:our-sjpc}. 
We present an analysis of our algorithm in terms of the estimation error and time and space bound in Section~\ref{sec:analysis}. Running time is analyzed in Section~\ref{sec:runtime}, and an extension of our algorithm for similarity join size estimation is studied in Section~\ref{sec:ext-join}.
Experimental results are reported in Section~\ref{sub:exp}, and related work is reviewed in
Section~\ref{sec:RelatedWork}. Section~\ref{sec:conclude} concludes the paper.

\section{Baseline Algorithms}
\label{sec:baseline-algs}



One widely used baseline is random sampling; it is 
applied in contexts similar to ours \cite{JoinSize99PODS,LNS11} 
and is very easy to implement. We also present the signature pattern counting~\cite{LNS09} and
the LSH-based bucketing 
~\cite{LNS11}, as two more baselines.

\subsection{Random sampling}
\label{sec:randomSampling}
For a set of $n$ records, 
one can pick $R$ different records uniformly 
at random (sampling without replacement); then
use the straightforward pair-wise comparison between the selected records 
to find $x_k$, $k=s,\ldots,d$ for the sample. Next, the  estimates can be scaled by the ratio of the size 
of the population~\footnote{Our population here is the set of record pairs being joined.} 
to the size of the sample space, i.e. $\frac{n(n-1)}{R(R-1)}$. Finally $g_s$ is estimated as in Eq.~\ref{eq:gsDef}.
For example, suppose random sampling selects rows $r1$ and $r3$ from Table~\ref{table:4rowex} with a sampling ratio of $0.5$. The selected rows are 2-similar, and the estimates of $x_3$, $x_2$ and $x_1$ for the sample are respectively 0, 2 and 0, and for the population are 0, $\frac{2*4*3}{2*1}=12$
and 0. 

\begin{table}[htb]
\caption{An example for $s$-similarity estimation}
\label{table:4rowex}
\begin{minipage}{.16\textwidth}
\begin{tabular}{l|lll}
{\bf $R$}           & A & B & C\\\hline
$r_1$ & $a_1$ & $b_1$ & $c_1$\\
$r_2$ & $a_2$ & $b_2$ & $c_2$\\
$r_3$ & $a_1$ & $b_1$ & $c_3$\\
$r_4$ & $a_3$ & $b_2$ & $c_2$
\end{tabular}
\end{minipage}%
\begin{minipage}{.85\textwidth}
\begin{tabular}{l}
3-similar pairs: \{($r_1$,$r_1$),($r_2$,$r_2$),($r_3$,$r_3$),($r_4$,$r_4$)\}\\
2-similar pairs: \{($r_1$,$r_3$),($r_3$,$r_1$),($r_2$,$r_4$),($r_4$,$r_2$)\}\\
1-similar pairs: \{\}\\
\end{tabular}
\end{minipage}
\end{table}

Alon et al.~\cite{JoinSize99PODS} used a similar random-sampling technique
in their experiments for estimating the self-join sizes of data streams.
However, the results show that it is not as accurate as other methods. Random sampling is also used in the literature for similarity join size estimation~\cite{LNS11}.

\begin{lemma} \label{lem:sampling_lower_bound}
Random-sampling requires a sample of size $\Omega(\sqrt{n})$ to give an 
estimate of the similarity self-join size with a relative error less
than $100\%$ with high probability. 
\end{lemma}
\begin{proof} 
\ifdefined\Archive
  See Appendix~\ref{sec:proofs}. 
\else
  See the extended version~\cite{RFarxiv2018}.
\fi
\end{proof}
\subsection{Signature pattern counting}
Lee et al.~\cite{LNS09} map the similarity self-join size estimation into the problem of finding a set of frequent signature patterns and estimating the number of records that match each pattern.
Each signature pattern is a record with some constants and some wildcard symbols say *. A data record matches a signature pattern if both have the same constants in respective columns; there is no matching constraint on columns marked with *. Clearly two records that match a signature pattern with {\it s} constants must be $k$-similar for $k\geq s$. Having a set of signature patterns each with at least {\it s} constants, and the number of records matching each pattern, one can estimate the similarity
self-join size for the set of tuples matching each pattern and add up the estimates. For example, this algorithm, applied to Table~\ref{table:4rowex}, can produce the patterns [a1,b1,*] and [*,b2,c2] both with frequencies $2$. 

However, this approach has some problems: (1) the estimate does not take into account
the overlap between patterns (e.g. [*,3,5,*] and [2,*,5,*]) which can lead to a double-counting; 
(2) the search space for patterns with frequency at least 2 is huge, and the number of such patterns
may not be much smaller than the size of the dataset. The authors address the first problem by 
placing the patterns in a lattice structure and estimating the size of the overlap between patterns for each level of the lattice. The second problem is addressed by modelling the pattern distribution as a power law and estimating (but not actually computing) the number of matches for low-frequency patterns based on the counts obtained for high-frequency ones.  

Their algorithm must (1) compute a set of frequent signature patterns and (2) collect the counts of records matching each pattern. 
A typical 
frequent counting is expected to make {\it d} passes over the data, where {\it d} is the record dimensionality, but one may consider either the partitioning scheme of Savasere et al.~\cite{SavasereON95} or the sampling method of Toivonen~\cite{Toivonen96} to cut the number of passes to two. Even Manku and Motwani~\cite{MankuM02} suggest a one-pass algorithm based on sticky sampling. However, these algorithms are generally useful in detecting very frequent patterns and this is reflected in their error bounds which is relative to the input length; they are likely to miss many less-frequent patterns. Also once the signature patterns are found, one more pass over data is needed to collect the actual counts. 

\subsection{LSH-based bucketing}
In this approach, Lee et al.~\cite{LNS11} map data records into some buckets using a locality sensitive hashing scheme. Two strata are defined for the sake of self-join sampling:
(1) record pairs that are mapped to the same bucket, (2) record pairs that
are mapped to two different buckets. Record pairs are sampled from each
stratum and their similarities are assessed. Finally, the results are 
scaled (based on record counts which are kept in each bucket) to derive
an estimate for the similarity join size.
%
To construct LSH buckets on disk, the algorithm has to read and write the whole data. 
One more pass is also needed to sample the record pairs.

\section{Our Estimation Algorithm}
\label{sec:our-sjpc}
The basic idea behind our estimation algorithm is to map the problem of similarity self-join size estimation into a set of self-join size estimations for which more efficient solutions are available, before putting together the results to obtain an estimate for the similarity self-join size. To illustrate the concept, consider Table $R$ with 3 columns and 4 rows, as shown in Table~\ref{table:4rowex}.

The self-join size of $R$ is $4$, and so is the number of records, hence $R$ has no 3-similar pairs other than self-pairs (as shown in Table~\ref{table:4rowex}). 
Now consider the projection of $R$ on columns (A,B), (A,C) and (B,C) with duplicates kept. This would yield tables $R_1$, $R_2$ and $R_3$, as shown with their self-join sizes.

\begin{center}
\begin{tabular}{ccc}
\begin{tabular}{lll}
{\bf $R_1$} & A & B\\\hline
&$a_1$ & $b_1$ \\
&$a_2$ & $b_2$ \\
&$a_1$ & $b_1$ \\
&$a_3$ & $b_2$ \\\hline
\multicolumn{3}{c}{s/join size=6}
\end{tabular} &
\begin{tabular}{lll}
{\bf $R_2$}&A & C\\\hline
&$a_1$ & $c_1$\\
&$a_2$ & $c_2$\\
&$a_1$ & $c_3$\\
&$a_3$ & $c_2$\\\hline
\multicolumn{3}{c}{s/join size=4}
\end{tabular} &
\begin{tabular}{lll}
{\bf $R_3$}&B & C\\\hline
&$b_1$ & $c_1$\\
&$b_2$ & $c_2$\\
&$b_1$ & $c_3$\\
&$b_2$ & $c_2$\\\hline
\multicolumn{3}{c}{s/join size=6}
\end{tabular} 
\end{tabular}\\
\end{center}

The self-join size of $R_1$ is 6, and that of $R$ is $4$. Excluding 3-similar pairs, $R$ must have $6-4=2$ pairs of rows that are 2-similar on columns (A,B). Similarly the self-join size of $R_3$ is 6, which indicates that $R$ has two pairs of records that are 2-similar on columns (B,C). No pairs of records in $R$ are 2-similar on columns (A,C). Putting the results together, one can conclude that $R$ has $2+2=4$ pairs of 2-similar records. This is  {\it the exact} number of 2-similar pairs, calculated solely based on the join sizes and the size of $R$. 



There are three problems associated with this approach: 
(1) the number of possible projections of a relation with $d$ attributes is $2^{d-1}$ and so is the number of self-join size estimates;
(2) the join sizes are not independent simply because
the projections of two $s$-similar records are expected to have some $s-1, s-2, \ldots, 1$-similar pairs;
(3) the sum of the number of records in the projected tables can be much larger than the input (or more precisely it can be larger by factor of $2^{d-1}$) and this has implications in the required space usage and per-record  processing time.

We address the first problem by reducing the number of self-join size estimations to $d-s+1$. This is done by collapsing all projections with $k$ attributes into a single stream. To make distinctions between tuples  from different projections though, we attach another column to the stream to indicate the projection from which the tuple is drawn. Otherwise, two tuples that have the same values under different projections can join, leading to wrong join sizes. Applying this to the projections with two attributes in our running example will yield a table with three columns, twelve rows and the self-join size of 16 (as shown in Table~\ref{table:sampling-ex}), of which 12 are self-pairs. This gives 16-12=4 pairs of 2-similar records, consistent with our previous results.


The second problem is addressed in the next section by calculating the contributions of an $s$-similar pair to the projections and incorporating it in our size estimations. We address the third problem through a combination of sampling and sketching, showing that both the space usage and the error can be bounded, and that the proposed algorithm outperforms our competitive baselines by a large margin.   

\subsection{Handling double-counting}
Given a table with $d$ attributes, the set of possible groupings of the attributes can be described in a lattice. 
Suppose the grouping that includes all attributes is at Level $d$ of the lattice; then level $d-1$ will have all subsets consisting of $d-1$ attributes, and so on  (as shown in Fig.~\ref{fig:lattice} with 4 attributes). 

\begin{figure}[ht]
\centering
\includegraphics[width=2.3in]{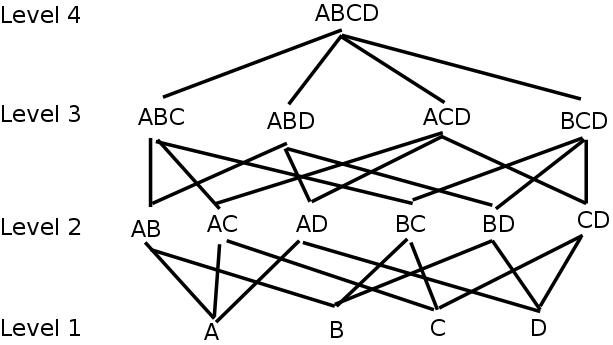} 
\caption{Attribute groupings for ABCD}
\label{fig:lattice}
\end{figure}  

To obtain a relationship between the self-join size and the number of similar pairs, let $y_k$ denote the self-join size at level $k$ of the lattice, and $x_k$ be the number of record pairs that are $k$-similar. 
\begin{lemma}
\label{lemma:XkYk-a}
Given $k \leq d$, for a set of $n$ records with no $k+1, k+2, \ldots, d$-similar pairs,
\[
x_k = y_k -{d \choose k}n.
\]
\end{lemma}
\begin{proof}
If two records are $k$-similar, then there must be a unique projection at level $k$ under which both records produce the same values; hence those values can join and will 
contribute to $y_k$. However, $y_k$ also includes self-pairs where a record joins itself. The number of those 
pairs is the same as the number of records at level $k$, which is ${d \choose k}n$. Subtracting the two will give $x_k$.
\end{proof}

Now we consider the scenario where the set can have $k+1, k+2, \ldots, d$-similar pairs.
Consider two records $r_1$ and $r_2$ that join at level $d$, meaning they have
the same values in all $d$ attributes. All projections of these records will also join, and this introduces double-counting in join-size estimates at levels $d-1, d-2, \ldots, 1$. 
The size of this projection can be precisely computed with not much effort though. 
\begin{lemma}
\label{lemma:XkYk-b}
For a set of $n$ records and $k \leq d$,
\begin{equation} 
\label{eq:prj-est}
x_k = y_k -{d \choose k}n -  \sum_{j=k+1}^d {j \choose k} x_j.
\end{equation}
\end{lemma}
\begin{proof}
Consider two records that are $j$-similar for $j>k$. Then there must be a projection at level $j$ under which the two records emit the same values; all projections of those values at level $k$ will be identical. There are 
${j \choose k}$ such projections. 
With $x_j$ giving the number of record pairs that are $j$-similar, the expression ${j \choose k} x_j$ gives the contribution of all $j$-similar pairs to $y_k$. The rest follows from Lemma~\ref{lemma:XkYk-a}.
\end{proof}



Unlike the approach of Lee et al.~\cite{LNS09} that computes the overlap between signature patterns, which is an approximation with no clear bound on the error, our method  computes the exact size of the overlap between projections.

\subsection{Sampling from the projections}
To calculate the number of pairs that are $s$-similar, we need the self-join sizes at levels $s$ to $d$ of the lattice. Each level of the lattice emits a stream that includes all record projections at that level. As discussed earlier in Section~\ref{sec:our-sjpc}, attaching a projection ordering to each row in this stream allows different projections at the same level all be collapsed into a single stream without introducing an estimation error, hence cutting the number of size estimations.
However, as shown in
Table~\ref{table:sampling-ex} for our running example, 
each row in our initial stream is listed under multiple
projections and all those projections 
contribute to our size estimation. 
Our objective in this section is to cut the size through sampling.

\begin{table}
\caption{Projections at levels 3 and 2}{
\begin{tabular}{cc}
Level 3 & Level 2\\
\begin{tabular}{llll}
Proj & \$1 & \$2 & \$3\\\hline
ABC&$a_1$ & $b_1$ & $c_1$ \\
ABC&$a_2$ & $b_2$ & $c_2$ \\
ABC&$a_1$ & $b_1$ & $c_3$ \\
ABC&$a_3$ & $b_2$ & $c_2$ \\ \hline
\multicolumn{4}{c}{s/join size=4}
\end{tabular} 
&
\begin{tabular}{lll}
Proj&\$1 & \$2\\\hline
AB& $a_1$ & $b_1$\\
AB& $a_2$ & $b_2$\\
AB& $a_1$ & $b_1$\\
AB& $a_3$ & $b_2$\\
AC& $a_1$ & $c_1$\\
AC&$a_2$ & $c_2$\\
AC&$a_1$ & $c_3$\\
AC&$a_3$ & $c_2$\\
BC& $b_1$ & $c_1$\\
BC& $b_2$ & $c_2$\\
BC& $b_1$ & $c_3$\\
BC& $b_2$ & $c_2$\\ \hline
\multicolumn{3}{c}{s/join size=16}
\end{tabular}\\
\end{tabular}}
\label{table:sampling-ex}
\end{table}

Let $0 < r \leq 1$ be our sampling ratio, meaning each row of an emitted stream is selected uniformly into the sample with probability $r$. This is sampling without replacement and is done by uniformly selecting at random $r {d \choose k}$ projections of each record at level $k$.
The sampling here is in the form of inclusion-exclusion (unlike the one discussed in  Section~\ref{sec:randomSampling}) and 
one does not need to store the sample to estimate the self-join size. 
For the same reason, the sample size can grow linearly with the input to avoid the estimation problem discussed in Section~\ref{sec:randomSampling}.


Given a sample as discussed above, let random variables $X_k$ and $X^\prime_k$ be  estimates of $x_k$ for the population and the sample respectively. Also let the random variable $Y_k$ be an estimate of $y_k$ for the sample.
The relationship between the expected number of $k$-similar pairs in the population, $X_k$, and
the self-join size of a sample from a $k$-value stream, $Y_k$, can be expressed as follows. 

\begin{lemma}
\label{lemma:XkYk-sample}
For a set of $n$ records, each of arity $d$, $k \leq d$ and sampling ratio $r$,
\begin{equation}
\label{eq:X_k}
X_k = (Y_k -r {d \choose k}  n) / r^2 -  \sum_{j=k+1}^d  {j \choose k} X_j .
\end{equation}
\end{lemma}
\begin{proof}
Let us initially assume that there are no $k+1, k+2, \ldots, d$-similar pairs. 
Given that each record is included with probability $r$, giving us a sample of size $r {d \choose k} n$, the relationship between 
$Y_k$ and $X^\prime_k$ can be expressed as
\[
Y_k = X^\prime_k + r {d \choose k} n.
\]
For a pair of $k-$similar records, the probability that they both make to the sample (and are counted in $X^\prime_k$) 
is $r^2$; hence $X^\prime_k = r^2 X_k$. Replacing this in the equation above, we get 
$Y_k = r^2 X_k + r {d \choose k} n$ and this can be rewritten as
\begin{equation}
\label{eq:X_kmid}
X_k = (Y_k -r {d \choose k}  n) / r^2.
\end{equation}
We can now relax our assumption and show by induction on $k$ that 
the statement of the lemma holds. 
The basis holds for $X_d$. Now suppose  the lemma holds for $X_{k+1},\ldots,X_d$,  meaning we can drive the values of 
$X_{k+1}, \ldots, X_d$ using Eq.~\ref{eq:X_k}. Then the contributions of $k+1, \ldots,d$-similar pairs toward $X_{k}$ can be computed as $\sum_{j=k+1}^d  {j \choose k} X_j$ (see the discussion in the proof of
Lemma~\ref{lemma:XkYk-b}). Subtracting this quantity from our earlier estimate in Eq.~\ref{eq:X_kmid} will give the final result, and this completes our proof.
\end{proof}



\subsection{The algorithm}

Our one-pass \textit{Similarity Self-Join Pair Count} (SJPC) method is depicted in Algorithm~\ref{algo:sjpc}.  
The algorithm can be broken down into three main steps.


\begin{algorithm}
\caption{SJPC algorithm}
\label{algo:sjpc}
\SetKwInOut{Input}{input}
\SetKwInOut{Output}{output}
\SetKwFunction{KwFn}{{\bf Procedure} f2toPairCnt}
\Input{Similarity threshold $s$, sampling ratio $r$, sketch width $wd$ and sketch depth $dp$}
\Input{A stream with $d$ columns} 
\Output{Similarity self-join size of the stream}
\BlankLine

Initialize $(d-s+1)$ sketches, each of width $wd$ and depth $dp$\;
\For{$k\leftarrow s$ to $d$}{
  colComb[$k$].size = $d \choose k$\;
  colComb[$k$].list = list of all $k$ column combinations of $d$ attributes\;
}
\For{each row in the stream}{
  \For{$k\leftarrow s$ to $d$}{
    sampleSize= colComb[$k$].size * r\;
    \If{sampleSize is not an integer}{
     round it up with probability $sampleSize-truncate(sampleSize)$ and round it down with the remaining probability\;
    }
    Let $C$ be $sampleSize$ entries from colComb[$k$].list chosen uniformly at random\;
    \For{each $c \in C$}{
      let $p$ be the projection of the row on column combination $c$\;
      $p$ = concat($c$,$p$)\; 
      fp = fingerprint($p$)\;
      sketch\_insert($k$, $fp$)\;
    }
  }
}
\For{$k\leftarrow s$ to $d$}{
  Y[k] = sketch\_estimateF2($k$)\;
}
Let $n$ be the size of the input stream (in terms of the number of rows)\;
f2toPairCnt(d,s,n,r,X,Y) \\
estimate = 0\;
\lFor{$k\leftarrow s$ to $d$}{
  estimate += X[k]
}  
return estimate\; 

\BlankLine
\KwFn{d,s,n,r,X,Y}\\
\For{$k\leftarrow d$ \emph{downto} s} {
  sampleSize= $d \choose k$ * r * n\;
  X[k] = Y[k] - sampleSize\;
  \For{$j\leftarrow k+1$ \emph{\KwTo} d} {
    X[k] -= $r^2$*$j \choose k$ * X[j]\;
  }
  X[k] = (X[k] $<$ 0) ? 0 : X[k];  // estimates cannot be negative
}
\lFor{$k\leftarrow s$ \emph{\KwTo} d} {
  X[k] /= $r^2$
}
\end{algorithm}

\textbf{Step 1: Generate projections and construct sketches (lines 1-20).} 
For each record and each $k=s, \dots, d$, 
the algorithm selects $k$ different attributes uniformly at random, and 
projects the record under these attributes (with duplicates kept);
the projected attribute values are encoded into a string along with the text of the attribute combination. 
We call this record 
a \textit{$k$-sub-value}, and the set of all $k$-sub-values at level $k$ a \textit{$k$-sub-value stream}.
For example, if the selected attributes for a row are A, B and C and their respective values are $a_1$, $b_1$ and $c_3$, then the generated $3$-sub-value will be $ABC.a_1.b_1.c_3$.
With this coding, all $k$-sub-values can be placed on the same stream and no two sub-values from different projections can join. This would reduce the number of self-join size estimations at level $k$ of the lattice from ${d \choose k}$ to one.
The process is repeated $l_k=r { d \choose k} $ times ($0 < r < 1$).

This step will produce $d-s+1$ sub-value streams, one for
each  $k=s \dots d$, and the number of $k$-sub-values in each sub-value stream is controlled by the sampling ratio $r$. Sub-values may be fingerprinted into more concise fixed length strings~\cite{broder93}, and a sketch may be constructed for each  sub-value stream instead of directly storing it. There are several sketching algorithms that estimate the self- join size in one pass~\cite{AMS99JCSS,SketchNet05VLDB}. We use Fast-AGMS~\cite{SketchNet05VLDB}, which  maintains $w$ counters (sketch width) and map elements in the stream
into one of those counters. Two $4$-universal hash functions $h_1$ and $h_2$ are used where
$h_1$ maps each element into
either $-1$ or $1$ and $h_2$ maps it into $[1, \ldots, w]$, both uniformly at random. For each incoming element $e$, the sketch is updated by
adding $h_1(e)$ to the counter at index $h_2(e)$.
Once the stream is processed, the self-join size is estimated by adding up the squares of all counter values.
In our case, $d-s+1$ sketches are needed to estimate the self-join sizes for that many sub-value streams.
To provide a better error bound, the process is often repeated $t$ times (sketch depth) and the median of those $t$ estimates are chosen. The sketch requires $t w$ counters to implement, and we are constructing $d-s+1$ such sketches for our estimation.

\textbf{Step 2: Find the self-join sizes (lines 21-23).}
The algorithm, finds the self-join size of each sub-value stream, using standard self-join size estimation methods~\cite{SketchNet05VLDB}.

\textbf{Step 3: Estimate the similarity self-join size (lines 24-28).}
With the self-join sizes $Y_k$ computed for $k=s, \ldots, d$ in the previous step, the similarity self-join size can be computed using Equation~\ref{eq:X_k}.

As an example, let $d=6$, $s=4$ and $r=0.5$ and suppose $Y_6$, $Y_5$ and $Y_4$ are computed; 
we can compute the similarity self-join size by adding up $X_4$, $X_5$ and $X_6$, where
the latter can be obtained by solving the following equation system:
\begin{equation}
\begin{cases}
  Y_6 = 0.25 X_6 + 0.5 n \\
  Y_5 = 1.5 X_6 + 0.25 X_5 + 3 n  \\
  Y_4 = 7.5 X_6 + 1.25 X_5 + 0.25 X_4 + 7.5 n.
\end{cases}
\end{equation}

Step 1 can be done while the input is being read, 
and Step 3 simply takes $d-s+1$ self-join size estimates and computes $X_k$ using Eq.~\ref{eq:X_k}, which is straightforward. This leaves us with the self-join size estimation in Step 2 for which one-pass algorithms are available. 

\section{Analysis}
\label{sec:analysis}
There are two sources of randomness in the proposed algorithm: (1) randomness due to the sampling in Step 1, and (2) randomness due to the self-join size estimation in Step 2. To get a better insight into the algorithm and its steps, we analyze it both without and with the randomness in Step 2. We refer to the case where an exact self-join size is computed in Step 2 as {\it the offline case}, and the case where this is estimated using a sketch as {\it the online case}.
  
\begin{theorem}\label{thm:E_VAR_offline}
(Unbiased estimate and variance - offline case)
The SJPC algorithm gives an unbiased estimate of the 
$s$-similarity self-join size under the offline scenario, i.e. $E[G_s] = g_s$, and
%
%
 $var[G_s/g_s]$ is at most 
$${d \choose s}^2\frac{1}{r} {2(d-s) \choose {d-s}} / g_s, $$
where $G_s$ is the estimate and $g_s$ is the true value.
\end{theorem}
\begin{proof}
\ifdefined\Archive
  See Appendix~\ref{sec:proofs}. 
\else
  See the extended version~\cite{RFarxiv2018}.
\fi
\end{proof}
\textit{Remarks.}
Since the estimate is unbiased, $var(\frac{G_s}{g_s})$ can be considered
as a measure of relative error in practice. There are a few observations that can be made.
First, this is an upper bound of the error and the actual error is expected to be less. More specifically, when $r=1$, the estimate has no error (see Lemma~\ref{lemma:XkYk-b}) whereas the bound can still be large depending on $d$ and $s$.
Second, the variance increases significantly as the gap between $s$ and $d$ widens. However, in many practical settings such as duplicate detection, often higher similarities (80\% or higher) are sought; in those cases, the error is expected to be low, as shown in Figure~\ref{fig:varianceOfflineOnline} (left) as well.
Third, when other parameters are fixed,
the expected relative error decreases when the true similarity self-join size increases.
Assuming $g_s$ increases quadratically with $n$ (which is the case in our real
dataset), the relative error goes down linearly with $n$. 
The results shown in the experiment section confirm this observation.

\begin{figure*}[tb]
\centering
\begin{minipage}{.3\textwidth}
\hspace*{-0.5cm}
\includegraphics[width=2.45in]{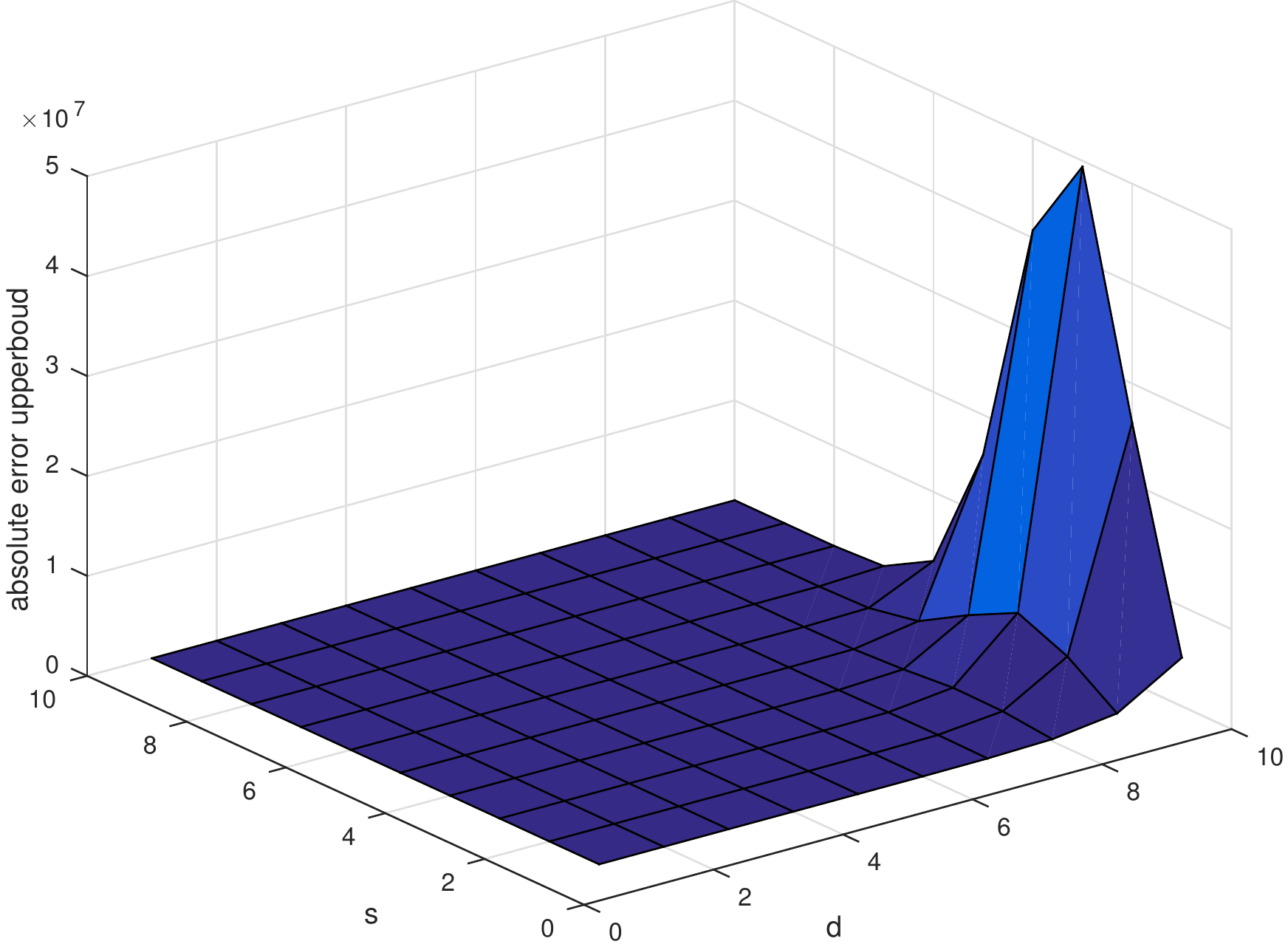}
\end{minipage}
\hspace*{0.25cm}
\begin{minipage}{.3\textwidth}
\hspace*{-0.5cm}
\includegraphics[width=2.45in]{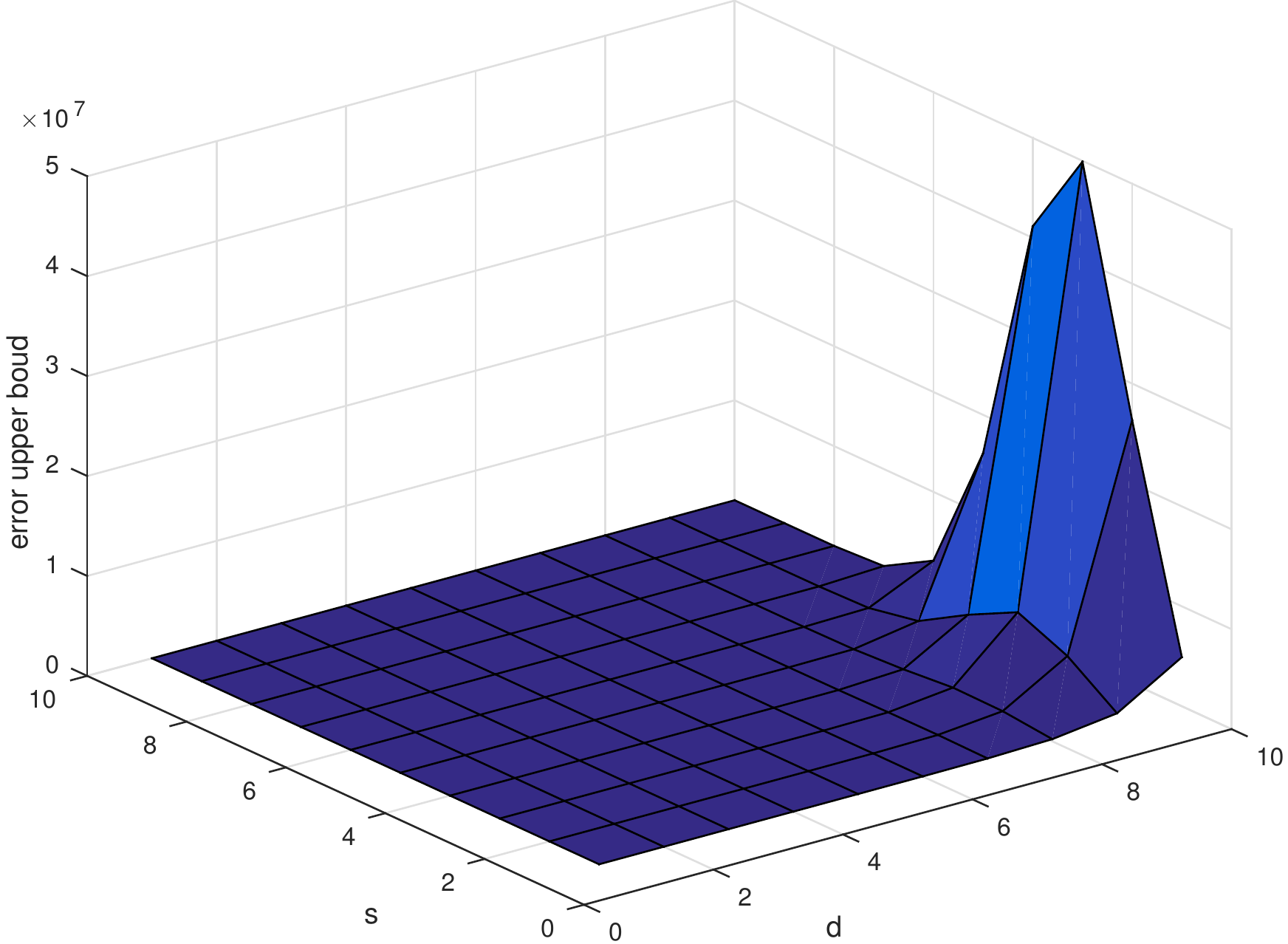}
\end{minipage}
\hspace*{0.25cm}
\begin{minipage}{.3\textwidth}
\hspace*{-0.5cm}
\includegraphics[width=2.45in]{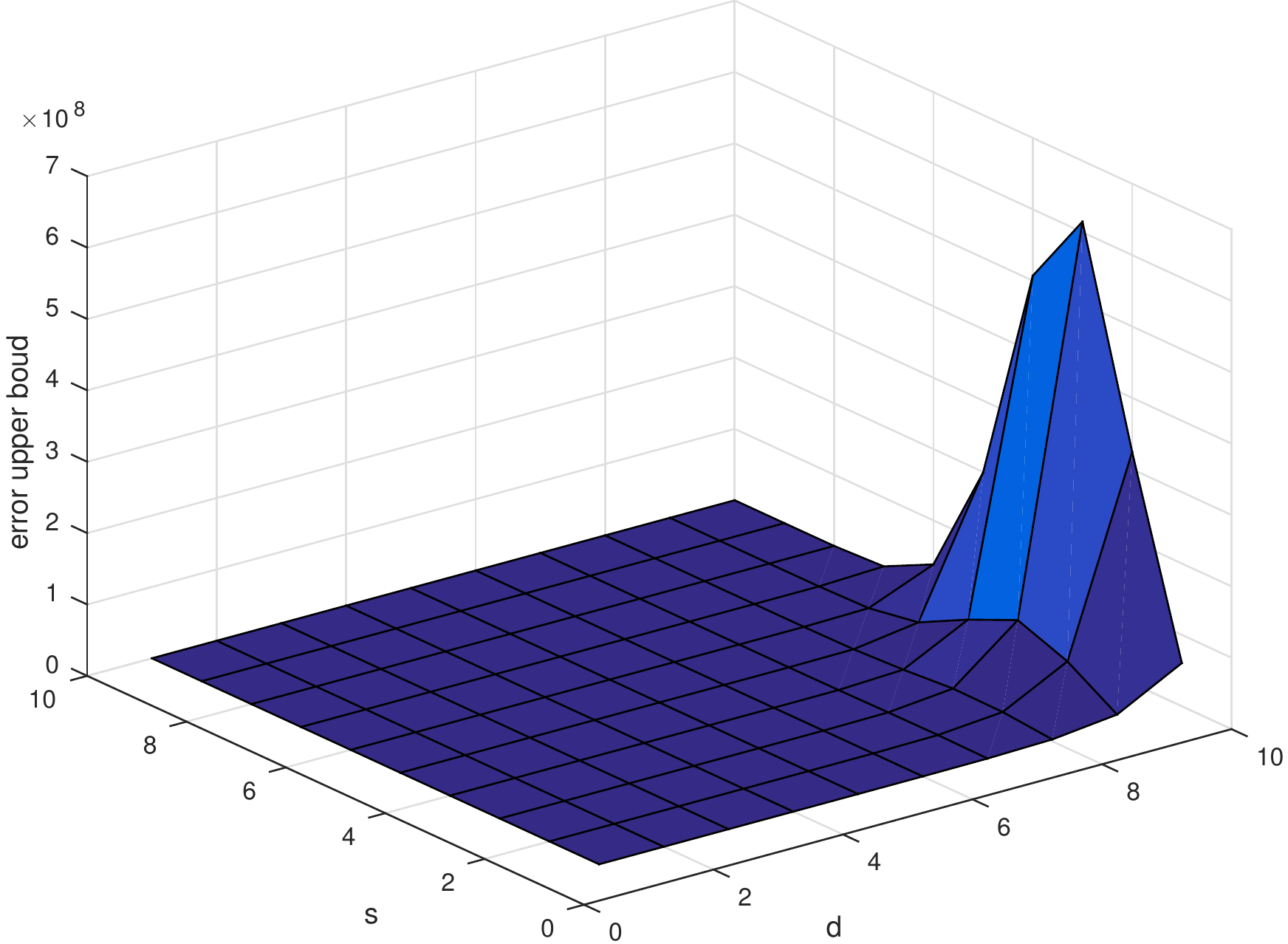}
\end{minipage}
\caption{Error upper bound (i.e. the absolute error $var(G_s)$) with gs=1 in (left) the offline case for r=1, (middle) the online case for r=1, w=1000, and (right) the online case for r=$0.1$, w=1000. Note the term $n/g_s$ in Theorem~\ref{thm:online-SJPC} (which is not larger than 1) is dropped to derive a weaker upper bound.}
\label{fig:varianceOfflineOnline}
\end{figure*}


When the dimensionality $d$ increases, the expected error increases and
time and space costs are also affected. Since the algorithm generates 
$r \sum_{k=s}^d {d \choose k}$ sub-values for each record, 
if the sampling ratio $r$ is chosen such that $\lceil r {d \choose k} \rceil < c$ for some constant $c$, 
then both the space and time costs in an offline case is
$O((d-s+1)n)$ for processing all records. Also for large $d$, it may be possible to select a subset of the columns and gauge the similarity based on the subset. 

Next we report the performance of our algorithm under an online scenario where the self-join size in Step 2 is estimated using a sketch.

\begin{theorem}\label{thm:online-SJPC}
(Unbiased estimate and variance)
The SJPC algorithm gives an unbiased estimate of the
$s$-similarity self-join size in an online scenario, i.e. $E[G_s] = g_s$, and the variance
of $\frac{G_s} {g_s}$ is at most
$$
{d \choose s}^2 \frac{1}{r}  {2(d-s) \choose d-s} 
((1+\frac{2}{w})/ g_s + \frac{2}{w} (1 + \frac{n}{r g_s})^2),
$$
%
where $w$ is the Fast-AGMS sketch width (depth is $1$), 
$d$ is the number of attributes, $s$ is the given similarity 
threshold, $r$ is the sampling ratio, $g_s$ is the true value
of the similarity self-join size, and $G_s$ is the estimated value.
\end{theorem}
\begin{proof}
\ifdefined\Archive
  See Appendix~\ref{sec:proofs}. 
\else
  See the extended version~\cite{RFarxiv2018}.
\fi
\end{proof}
%

\textit{Remarks.}
A few observations can be made here. First, as in the offline case, this is an upper bound of the error. In particular, when the sampling ratio is close to $1$, the offline estimates are expected to be accurate and the only source of error is from $d-s+1$ sketches.
Second, the variance gets a hit as the gap between $d$ and $s$  widens or $d$ becomes large. 
Again, this is not an issue in many practical settings where a high similarity (80\% or higher) is desired.
Third, to bound the variance, the space usage (denoted by $w$) does not have to increase when $n$ increases as long as
$g_s$ increases proportionally, which is usually an expected case.
Finally, the statement provides a formulation of the interaction between sketching and sampling, and how the variance changes with $r$ (see the right two columns of Figure~\ref{fig:varianceOfflineOnline}). A similar (but more extensive) study on the interaction between sketching and sampling in a different context is conducted by Rusu and Dobra~\cite{RusuD09}, where they reach the same conclusion that sketching over samples is a viable option, reducing the processing time with not much loss in accuracy.

\begin{theorem}\label{thm:selectivity_online}
(Space and time cost to bound the selectivity estimation error)
The SJPC algorithm guarantees that the
estimated selectivity of the similarity self-join deviates from the true value
by at most $\epsilon$ with probability at least $1-\lambda$. More precisely,
$Pr[|\hat{\theta}_s - \theta_s| \leq \epsilon] \ge 1-\lambda$, where
$\hat{\theta}_s$ is the estimated selectivity and $\theta_s$ is the true
value. The space cost is $O(\log(1/\lambda) (d-s+1) w)$,
and the time cost for processing each record is
$O(\log(1/\lambda) {d \choose s}^2 {2(d-s) \choose d-s} 
(\sum_{k=s}^d {d \choose k}) / (\epsilon^2 w)) $.
\end{theorem}
\begin{proof} 
\ifdefined\Archive
  See Appendix~\ref{sec:proofs}. 
\else
  See the extended version~\cite{RFarxiv2018}.
\fi
\end{proof}
%
\textit{Remarks.}
Note that $g_s$ appears neither in the time nor in the space complexity.
Although this statement discusses
the error of selectivity estimation, which is a relative error based on
$n^2$, by slightly changing the proof, it is not hard to see 
the statement also holds for relative errors defined based on $g_s$.
The algorithm constructs $d-s+1$ sketches each of size $O(\frac{1}{\epsilon^2} \log(1/\delta))$, giving a space cost of $O((d-s+1)\frac{1}{\epsilon^2} \log(1/\delta))$, meaning that
using constant time per record and constant number of counters, the algorithm can give accurate 
estimates of similarity self-join size with high probability.
It should be noted though that each counter needs $\log F$ bits to implement, where $F$ is the maximum frequency of a sub-value. In the extreme case where records all have the same value, $F$ would be $O(n)$.

The statement also shows that both time and space costs will increase
when $d$ increases or the gap between $d$ and $s$ widens. 
There is a clear tradeoff between time and space
controlled by $w$ (implicitly by $r$). 
If $w$ is large, time cost will be smaller while 
space cost will be larger. 
Compared with the offline case, the online case requires much less space and returns a final
estimate much faster after scanning the dataset once. Our experiments show that the overall error
in the online case is still negligible and much less than the competitors under typical settings of $d$ and $s$. 

\section{Asymptotic Time Compared to Random Sampling}
\label{sec:runtime}
In terms of a time comparison, there are two main stages in both SJPC and random sampling: data summarization and size calculation. 
At the data summarization stage, random-sampling takes $O(1)$ time per record, whereas SJPC has to construct $r \sum_{k=s}^d {d \choose k}$ sub-values for each record, and for each sub-value $t$ counters will be updated, where $r$ is the sampling ratio and $t$ is the sketch depth. Thus SJPC
will take $O(rt \sum_{k=s}^d {d \choose k})$ time per record. Random sampling is clearly faster at the data summarization stage.

At the size calculation stage, having the data summary, SJPC
has to compute the mean or median of $(d-s+1) w t$  counters and plug them in Equation~\ref{eq:X_k}
to find the estimates. Hence it will take $O( (d-s+1) w t)$ time,
which basically is the time for scanning the data summary once, while random
sampling takes $O(R(R-1)/2)$ time for a sample of size R, quadratic in the sample size. 
Thus, online SJPC is faster at the size calculation stage. 

The total time $T(n)$ for SJPC is $O( n r t \sum_{k=s}^d {d \choose k} + (d-s+1)wt)$, which can be written as $O(n \sum_{k=s}^d {d \choose k})$, with the sampling ratio $r$ and the sketch depth $t$ treated as constants. Also since $\sum_{k=s}^d {d \choose k} \leq d^{d-s}$, we can use $O(n d^{d-s})$ as an upper bound of the running time.

The cost of random sampling depends on the sample size, which is a function of $n$. We know from 
Lemma~\ref{lem:sampling_lower_bound} that the size of the sample must be larger than $n^{1/2}$ to obtain an estimate that has an error less than 100\%. Let the sample size be $n^\frac{p}{p+1}$ for $p>1$. Random sampling needs $T(n)=O(n^\frac{2p}{p+1})$ time to obtain an estimate. Figure~\ref{fig:time-rs-sj-varyd} shows how the two methods cope as the dataset size $n$ and the dimensionality $d$ increase. As expected, random sampling suffers when $n$ increases whereas SJPC suffers when $d$ increases or $d-s$ widens. On the other hand, SJPC scales linearly with $n$.

\begin{figure*}[ht]
 \begin{center}
   \includegraphics[width=3.0in]{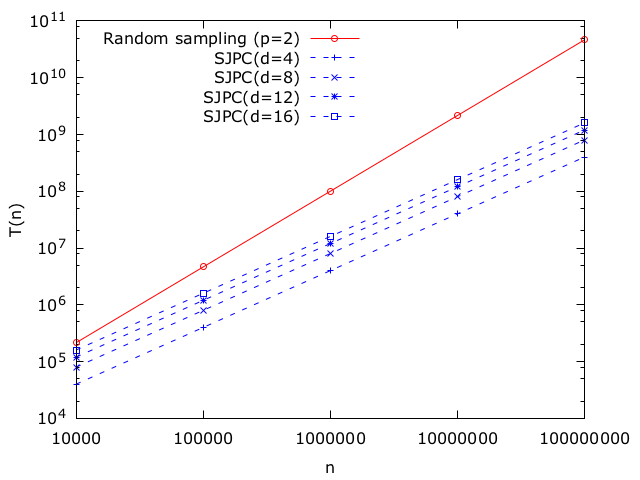} 
   \includegraphics[width=3.0in]{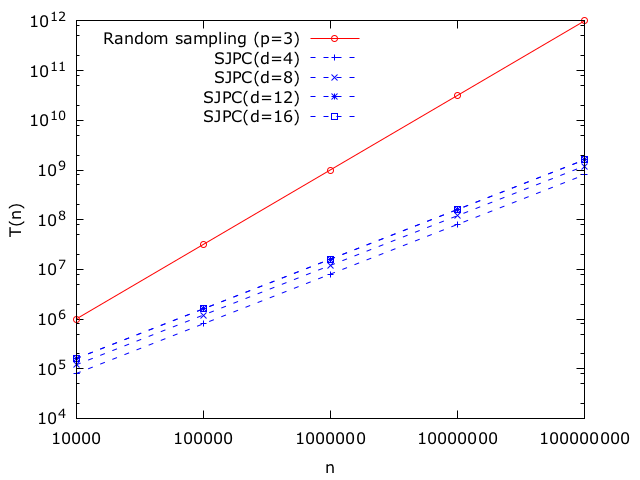}
    \caption{Asymptotic time varying $n$ and $d$, with $p=2$ and $s=d-1$ on the left and $p=3$ and $s=d-2$ on the right}
    \label{fig:time-rs-sj-varyd}
  \end{center} 
\end{figure*}

\section{Similarity Join Size Estimation}
\label{sec:ext-join}
An important problem related to similarity self-join size estimation is estimating similarity join sizes. 
First we show an estimate that does not hold for similarity join sizes.  A simple estimate for join size is based on the self-join sizes. Alon et al.~\cite{JoinSize99PODS} show that for two relations $A$ and $B$
\[
|A \bowtie B| \leq \frac{SJ(A) + SJ(B)}{2}
\]
where $SJ(A)$ and $SJ(B)$ are self-join sizes of $A$ and $B$ respectively on the joining attributes. This does not hold for similarity join though. Here is a counter example. Let $A$ consists of the row $(a,b,c,d)$ and $B$ consists of $(a,b,cx,dx)$ and $(ax,bx,c,d)$. With the similarity threshold set to $2$, rows in $A$ and $B$ join, and the join size is $2$; but the similarity self-join size of $A$ is $1$ and that of $B$ is $2$ and the bound on the join size does not hold. The same can be shown for larger thresholds. For example, with a similarity threshold set to $3$, let $A$ be the same as above and $B$ be the set of three rows $(ax,b,c,d)$, $(a,bx,c,d)$, and $(a,b,cx,d)$. Again the bound does not hold.

A well-known fact in both join and self-join size estimation is that an estimation is generally ineffective when the size to be estimated is small compared to the sizes of the relations being joined,  and a sanity check may be performed to avoid such cases~\cite{JoinSize99PODS}.
Consider the problem of similarity join size estimation between two relations $R$ and $S$ in the presence of one such sanity check. 
An estimation
algorithm may look like this: (1) project the records of $A$ and $B$ independently into sub-value streams (as discussed in Sec.~\ref{sec:our-sjpc}), (2) construct a sketch for each sub-value stream for a total of $2(d-s+1)$ sketches,  (3) estimate the join sizes between sub-value streams of $A$ and $B$ at each of the levels $s,\ldots,d$, and (4) estimate the join size based on the join sizes of the sub-value streams. As discussed for self-join sizes, the computation in Step 4 is exact meaning given exact join sizes of the sub-value streams, no error can be introduced in Step 4.
The only source of error here is (a) error due to sampling from projections in Steps 1-2. and (b) the sketch error in estimating join sizes between sub-value streams.
   
The join size estimation in Step 3 uses the product of the sketches for sub-value streams; in particular,   
given (AGMS and Fast-AGMS) sketches $S(A)$ and $S(B)$ of
two relations $A$ and $B$ respectively, an estimator for $A \bowtie B$ is $S(A).S(B)$. 
It is easy to show that this estimate is unbiased since the expected contribution of non-matching values to the product is zero when the sketch mapping functions are 2-wise (in case of AGMS) or 4-wise (in case of Fast-AGMS) independent.
Alon et al.~\cite{JoinSize99PODS}
show that this estimate has a variance which does not exceed two times the product of self-join sizes of $A$ and $B$. 

Let random variables $X_k$ and $Y_k$ denote respectively the similarity join size and the join size both at level $k$. The relationship between the two variables can be written as
\begin{equation}
\label{eq:X_kJ}
X_k =Y_k/r^2 - \sum_{j=k+1}^d {j \choose k} X_j,
\end{equation}
where $r$ is the sampling rate, set to the same value for both streams.
Note that in case of a similarity join size estimation, there is no self-pair (where a record joins itself) and this gives rise to the slight difference between this estimate and that in Equation~\ref{eq:X_k}.

The estimate is unbiased since $Y_k$ is unbiased and has the expectation
\begin{equation}
\label{eq:E-X_kJ}
E(X_k) = E(Y_k)/r^2 - \sum_{j=k+1}^d {j \choose k} E(X_j)
\end{equation}
and the variance
\begin{align*}
Var(X_k) &= Var(Y_k)/r^4 + \sum_{j=k+1}^d {j \choose k}^2 Var(X_j) \\
 & -\sum_{j=k+1}^d (2/r^2) {j \choose k} Cov(Y_k, X_j)
\end{align*}
which can be bounded as
\begin{equation}
\label{eq:var-X_kJbound}
Var(X_k) \leq Var(Y_k)/r^4 + \sum_{j=k+1}^d {j \choose k}^2 Var(X_j).
\end{equation}

\section{Experiments}\label{sub:exp}
To verify our analytical findings in more practical settings, and to assess both the robustness and the performance of the SJPC algorithm, we conducted a set of experiments on both real and synthetic data under different settings including different similarity thresholds, dataset sizes, and dimensionalities. 
%
When applicable, the performance of our method is compared to that of the competitors including the LSH-based bucketing and random sampling (see Sec.~\ref{sec:baseline-algs} for details of these algorithms).

\subsection{Experimental Setup}
The following three datasets were used in our evaluation (see also Section~\ref{sec:runtimeExp} for larger datasets and experiments).\\
\noindent
{\bf DBLP5.} This was a set of records selected from DBLP~\footnote{http://dblp.uni-trier.de/xml}. The selection criteria was that a record was selected if it had non-empty values in (all of) the following 5 fields: {\it title, author, journal, volume} and {\it year}. In the first 20,000 records that were qualified, there were 19884 unique titles, 15917 unique authors, 29 unique journals, 125 unique volumes and 49 unique years.

\noindent
{\bf DBLP6.} This was similar to DBLP5 except every record here had non-empty values in the following 6 fields: {\it title, author, journal, month, year} and {\it volume}. The dataset had 2468 records. There were 2456 unique titles, 1601 unique authors, 9 unique journals, 150 unique volumes, 41 unique years and 26 unique months. 

\noindent
{\bf DBLPtitles.} This was a set of paper titles from DBLP with each title fingerprinted into 6 super-shingles where each super-shingle was a 64 bit fingerprint. This resembled the experimental setting of Henzinger~\cite{nearDuplicate06Henzinger} and Broder et al.~\cite{BroderGMZ97ComputNet}, where their goal was to find near-duplicate Web pages. This resulted in 467,468 records, each with 6 attributes. The number of unique values in each column ranged from 27000 to 30000.

\begin{table}
\caption{Accumulative count of $s$-similar pairs, excluding self-pairs. The count shown at row $s=i$ includes all pairs that have a $s$-similarity $i$ or larger.} 
\label{table:dataStat}
\begin{tabular}{r|rrr}
 & DBLP5 & DBLP6 & DBLPtitles\\ \cline{2-4}
\backslashbox{s}{n} & 20,000 & 2468 & 200,000\\ \hline
6 & na               & 0               & 19,356\\
5 & 70               & 26             & 210,666\\
4 & 761             & 7,984        & 1,900,702\\
3 & 1,827,680   & 29,405      & 16607104\\
2 & 2,112,300   & 184,287    & 103,992,978\\
1 & 39,556,445 & 1,655,537 & 521,423,328\\
\end{tabular}
\end{table}


Table~\ref{table:dataStat} gives more stats on these datasets including similarity join sizes for different similarity thresholds.
Unless stated otherwise, all experiments are repeated 30 times and either the mean, the standard deviation or both of the relative error is reported.

\subsection{Offline scenario}
In the first set of our experiments, we wanted (1) to evaluate our method without introducing any error due to the sketching and (2) to characterize its performance compared to other baselines. This was possible in the offline scenario (as discussed in Section~\ref{sec:analysis}), under which the baseline algorithms introduced in Section~\ref{sec:baseline-algs} could as well be applied.
Next we compare the performance of our method to some of these non-streaming solutions under the same or similar space requirements.

\textbf{A note on the signature pattern counting of Lee et al.~\cite{LNS09}}.
We think there is a mistake in the formulation presented in the paper. In particular, with the formulation of $C_{l,t}$ in the authors' Equation~4, the estimates of similarity self-join size can be negative. This is what we observed in our experiments of running this algorithm on DBLP5 and DBLP6. Also the estimates were sometimes off by a factor of 2 or larger. We carefully verified our implementation and it was indeed consistent with the paper. We also noticed that Equation~4 applied to the authors' own example of LC(2) on Page~6 would give $-2$ instead of the reported result $6$. After communicating this with the authors and given the fact that the same authors show LSH-SS outperforms the signature pattern counting, we decided not to report our results for the latter. 

\textbf{Relative error on DBLP5, DBLP6 and DBLPtitles}.
In this experiment, we compare the performance of our method to LSH-based bucketing of Lee et al.~\cite{LNS11}; the selected algorithm for LSH-based bucketing is referred to as LSH-SS by the authors, which is shown to perform the best in their experiments. The sampling ratio for SJPC was set to 0.5 and $m_H$ and $m_L$ for LSH-SS was set to $n$, the size of the dataset, as suggested by the authors.

Figures~\ref{fig:err2metdblp6} shows both the mean and the standard deviation (std) of the relative error over 30 runs on DBLP6. In terms of both the mean and the standard deviation of the error, SJPC outperforms LSH-SS and has a standard deviation of the error which is sometimes an order of magnitude smaller than that of LSH-SS.  
The dataset had no 6-similar pairs, and both algorithms detected that correctly. 
Similar results are observed on DBLP titles as shown in Figures~\ref{fig:err2metdblpti}. In another experiment, we evaluated LSH-SS under two different sampling strategies. In the first strategy, referred to LSH-SSv1, the sampling ratio was set as suggested by the authors, i.e. $m_H=m_L=n$, and the sample size grew linearly with $n$. In our second strategy, we set the sampling ratio to a constant (set to $0.005$ in our experiments), meaning each pair was sampled with a fixed probability and the sample size grew linearly with the number of pairs. The results on DBLP5 are shown in Figure~\ref{fig:err2metdblp5}.
 
\begin{figure*}[tb]
\centering
\begin{minipage}{.3\textwidth}
\hspace*{-0.5cm}
  \centering\includegraphics[width=2.35in]{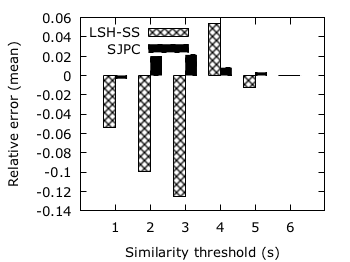} 
\end{minipage}
\hspace*{0.3cm}
\begin{minipage}{.3\textwidth}
\hspace*{-0.5cm}
  \centering\includegraphics[width=2.35in]{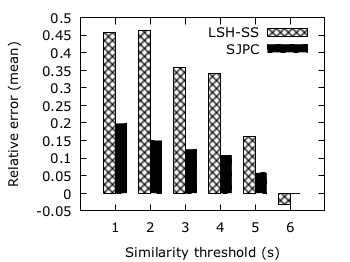} 
\end{minipage}
\hspace*{0.3cm}
\begin{minipage}{.3\textwidth}
\hspace*{-0.5cm}
  \centering\includegraphics[width=2.35in]{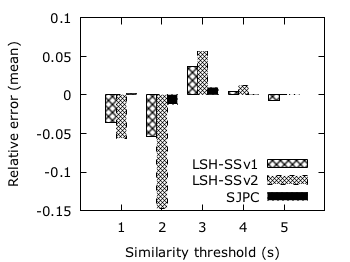} 
\end{minipage}
\begin{minipage}{.3\textwidth}
\hspace*{-0.5cm}
  \centering\includegraphics[width=2.35in]{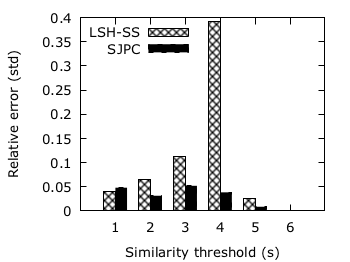}
  \caption{Relative error comparison on DBLP6 (offline case)}
  \label{fig:err2metdblp6}
\end{minipage}
\hspace*{0.3cm}
\begin{minipage}{.3\textwidth}
\hspace*{-0.5cm}
  \centering\includegraphics[width=2.35in]{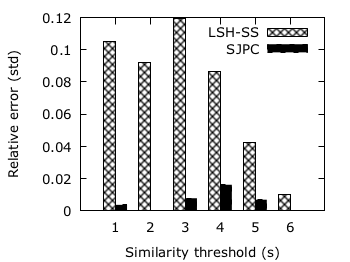}
  \caption{Relative error comparison on DBLPtitles (offline case)}
  \label{fig:err2metdblpti}
\end{minipage}
\hspace*{0.3cm}
\begin{minipage}{.3\textwidth}
\hspace*{-0.5cm}
  \centering\includegraphics[width=2.35in]{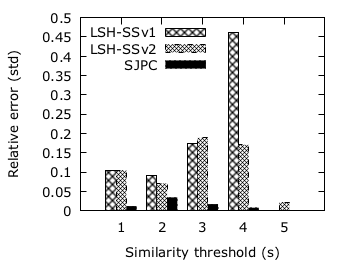}
  \caption{Relative error comparison on DBLP5 (offline case)}
  \label{fig:err2metdblp5}
\end{minipage}
\end{figure*}

\textbf{Materializing sub-value streams}.
Limiting our algorithm into the offline case (i.e. without the space and time-cost optimization due to the sketching) allowed us to compare its performance to multi-pass algorithms that assume the dataset and/or the intermediate data structures can be materialized. This is not usually feasible in a streaming environment and is not the right setting for our algorithm. That said, the offline case can be executed if the intermediate sub-value streams can be materialized. This is what we did in an 
implementation of both SJPC-offline and LSH-SS, where the memory usage for each method was tracked at various points during the execution (e.g. when a variable is defined or loaded) by calling a task manager function, before and after, and the largest difference for each method was reported. We verified the accuracy of this method by loading datasets of known sizes and comparing the space usage reported using this method with the actual size, and the method was accurate in the range of Kilobytes especially if the experiment was repeated.
As shown in Figure~\ref{fig:mem2metdblp5}, the space needed for materializing sub-value streams, to our surprise, was not much more than that of LSH-SS especially for large similarity thresholds (which is usually the case in similarity estimations), and this makes SJPC-offline a viable option due to its better error bounds.
%
%
\begin{figure}[ht]
 \begin{center}
   \includegraphics[width=2.5in]{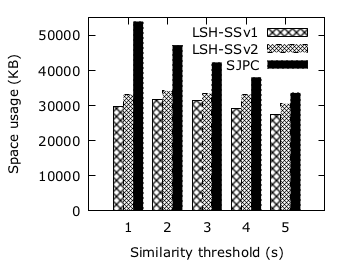} 
    \caption{Materialization cost on DBLP5}
    \label{fig:mem2metdblp5}
  \end{center} 
\end{figure}

\subsection{Online scenario}
\label{sec:exp-online}

In an online scenario where only one pass can be made over the data, random sampling is the only competitor. In this section, we compare the accuracy of SJPC to random sampling.

\textbf{Comparison to random sampling}.
Similar to the offline scenario, 
we set the sampling ratio to $0.5$, and ran our online SJPC 
on the first $200K$ rows of DBLPtitles. 
The sketch width (number of counters) was set to $1000$, 
and the sketch depth was set to $3$. SJPC needs one sketch for every sub-value stream, and the number of
sub-value streams is $d-s+1$ where $d$ is the data dimensionally and $s$ is the minimum similarity threshold that is desired. One can cover all useful similarity ranges (e.g. $s=3,\ldots,6$) by creating $4$ sketches on this particular dataset; this translates to 12,000 counters, each implemented as a 32-bit integer, giving a total space of 48,000 bytes. The same amount of space was allocated to random sampling. Every record of DBLPtitles had 6 fields, and each field was a 64-bit fingerprint, adding up to $48$ bytes per record. That meant, random sampling were given space for $1000$ records.


Both random sampling and SJPC give unbiased  estimates, hence we compare their standard deviations of the estimates.
As shown in Figure~\ref{fig:err2metdblptitles}, SJPC outperforms random sampling by a large margin. 
The standard deviation of the estimates for random sampling is almost an order of magnitude higher.
Also it should be noted that this was under the setting that the sketches are maintained for all similarity thresholds $s=3,\ldots,6$. For example,
with larger values of $s$, the space usage of SJPC is reduced (in terms of the number of counters to the number of records) while the accuracy remains the same;  this cannot be done in random sampling without affecting its accuracy.
\begin{figure}[ht]
 \begin{center}
   \includegraphics[width=2.5in]{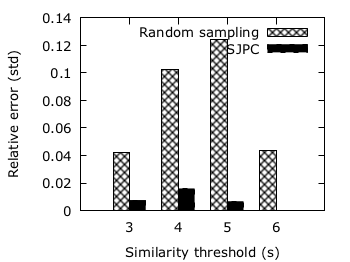} 
    \caption{Relative error on DBLPtitles (online case)}
    \label{fig:err2metdblptitles}
  \end{center} 
\end{figure}

Records in both DBLP5 and DBLP6 were longer and a bigger difference in performance between SJPC and Random sampling was expected. In particular, the same setting of our sketching could be used on DBLP5 and DBLP6 since a 32-bit counter was enough to keep the counts. However, random sampling suffered because of the space limitations. The average length of a record in DBLP5 was 167 characters and in DBLP6 was 121 characters. Under ASCII encoding where each character takes one byte, random sampling would have enough space to store 287 records of DBLP5 and 397 records of DBLP6. For the same reason, the results are not reported on these two datasets.

\subsection{Varying the parameters of SJPC}
In this section, we evaluate the performance of the SJPC algorithm under different parameter settings.

\textbf{Varying the sampling ratio}.
In all our previous experiments, the sampling ratio was set to $0.5$ meaning only half of the sub-values were sampled. The sampling ratio only affects the per-record processing time and not the space usage, hence if the processing time is not a constraint, the sampling ratio should be $1$ to obtain a better estimate. To study the relationship between the sampling ratio and the accuracy, we varied the sampling ratio from $0.25$ to $1$ while keeping everything else the same as before, i.e. 200K records of DBLPtitles with the sketch width and depth set at 1000 and 3 respectively. Figure~\ref{fig:err-dblptitles-varyf-d-s} (left) shows the effect of the sampling ratio on the standard deviation of the error. The error consistently drops as the sampling rate increases with an exception at the similarity threshold 1 where a sampling ratio of 0.5 performs slightly better than the next sampling ratio. We don't have a good explanation here other than confirming that this is due to the interaction between sketching and sampling and that the sketch in this particular case performed better on the sample. A similar behaviour is observed by Rusu and Dorba~\cite{RusuD09} in some of their experiments on constructing sketches over samples.
The mean error also drops (not shown here) but the drop is not as significant as the drop in the standard deviation. 
An observation that can be made is that the sampling ratio can vary between sub-value streams, for example, to reduce the error at certain values of $k$. It is easy to incorporate this in our formulation in Equation~\ref{eq:X_k}.

\begin{figure*}[tb]
\centering
\begin{minipage}{.3\textwidth}
\hspace*{-0.5cm}
\includegraphics[width=2.35in]{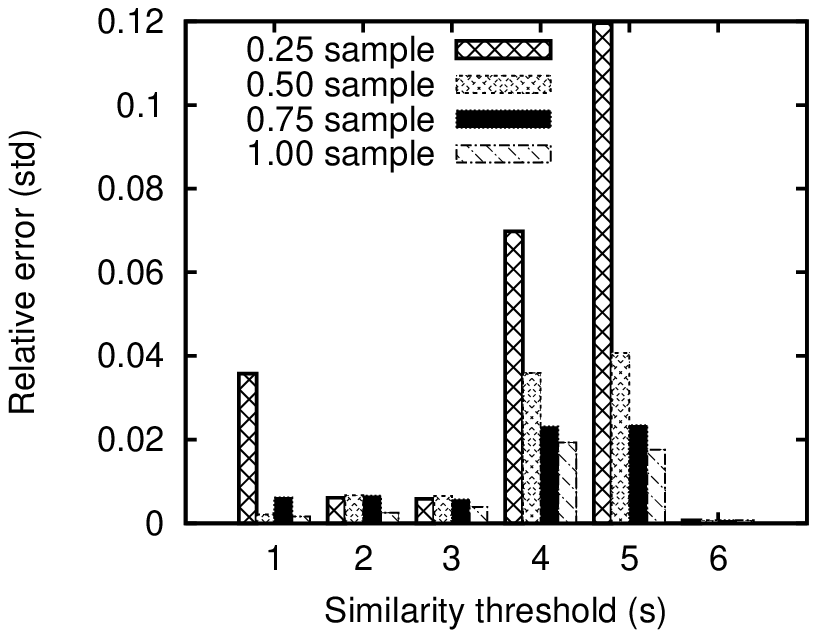}
\end{minipage}
\hspace*{0.25cm}
\begin{minipage}{.3\textwidth}
\hspace*{-0.5cm}
\includegraphics[width=2.35in]{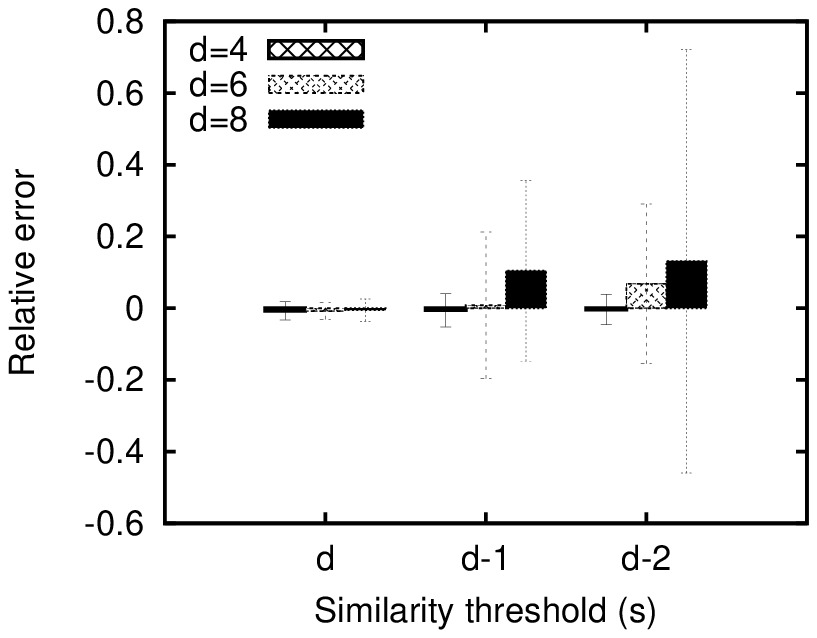}
\end{minipage}
\hspace*{0.25cm}
\begin{minipage}{.3\textwidth}
\hspace*{-0.5cm}
\includegraphics[width=2.35in]{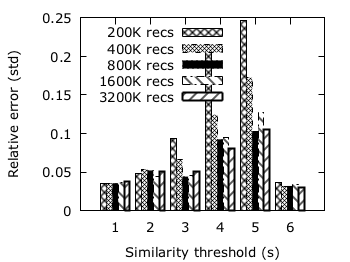}
\end{minipage}
\caption{Relative error std on DBLPtitles varying the sampling ratio (left), the number of columns (middle), and dataset size (right)}
\label{fig:err-dblptitles-varyf-d-s}
\end{figure*}


\textbf{Varying the dimensionality}.
In another experiment, we used the first 100K records of DBLPtitles but varied the number of columns by generating different number of super-shingles for each text. Other parameters were kept the same (i.e. sampling ratio of $0.5$ and sketch width at $1000$ and sketch depth at $3$). 
As shown in Figure~\ref{fig:err-dblptitles-varyf-d-s} (middle), the standard deviation of the error takes a hit when the dimensionality increases, which is consistent with our analytical prediction. This is under the condition that all other parameters are kept the same. Clearly one can reduce the error by increasing the sketch width/depth and/or the sampling ratio, as shown in our previous experiments.


\textbf{Varying dataset size and the number of duplicates}.
In this set of experiments, we started with 400K records of DBLPtitles and duplicated each record $X$ times with $X$ taking the values $1$, $2$, $4$ and $8$. This gave us datasets of sizes 400K, 800K, 1600K and 3200K. This particular setting allowed us to easily compute the true sizes without doing a join on larger files. We also included the first 200K records of DBLPtitles to see if the trend is the same when records are not duplicates.
As before, the sampling ratio of our method was set to $0.5$, and sketch
width and depth to $1000$ and $3$ respectively. 
Figure~\ref{fig:err-dblptitles-varyf-d-s} (right) shows that SJPC does not suffer when the dataset size increases while keeping the space usage and the sampling ratio the same; in fact for some similarity thresholds (e.g. 3, 4 and 5), the error drops as the dataset size increases. Compare this to random sampling where the sample size should increase at least as a square root function of the input size to maintain the same error rate. 
In contrast, having a larger number
of records can even be helpful for the sampling part of our algorithm, as hinted in our analytical results and also revealed in Figure~\ref{fig:err-dblptitles-varyf-d-s}.



\subsection{Running time} 
\label{sec:runtimeExp}
We conducted experiments to evaluate the running time of SJPC and its scalability with the dataset 
size, compared to random sampling.
The evaluation was conducted on larger datasets (more precisely, orders of magnitude larger than those used in earlier sections), as discussed next. One dataset was real, and three dataset were generated synthetically with varying degrees of skew to show how the skew may affect the scalability. Each records in both real and synthetic data consisted of 
5 columns~\footnote{A synthetic data record with 5 columns may represent, for example, papers with fields such as {\it first author, second author} (if any), {\it title, year,} and {\it venue}; a 4-similar pair in this case can be two copies of the same paper with one mismatched field. }. 

\noindent
{\bf Near-uniform 40-60.} This is a set of randomly generated 5-fields records with each field formed by concatenating two long integers (making a 64 bits field). 40\% of the records are unique, and each of the remaining 60\% have one $4$-similar pair. 

\noindent
{\bf Skewed 20-80.} This is a set of randomly generated records with each field formed by concatenating two long integers. 20\% of the records are unique and each of the remaining 80\% have 15 $4$-similar pairs. If each set of similar records is treated as an entity, then 20\% of the entities make up 80\% of the records.
 
 \noindent
{\bf Skewed 10-90.} Similar to Skewed 20-80, 10\% of the entities (each described by a set of similar records) make up 90\% of the records.

 \noindent
 {\bf YFCC.}  This is a set of 21 million records from Flickr 100 million photo dataset~\cite{thomee2016yfcc100m}. Each record in our case includes the following 5 fields: userid, date taken, the capturing device, the latitude and the longitude.

The experiments were conducted on a machine with AMD quad-core 2.3 GHz CPU and 16GB ram running Ubuntu.
Both algorithms SJPC and random sampling were implemented in {\em C} programming language and were compiled using {\em gcc}. 

For SJPC,  the space usage was fixed for all datasets, with the sketch depth and width respectively set to 1000 and 3 as before and the sampling ratio $r=1$.
For random sampling, we varied the sample size until random sampling could catch up SJPC in terms of the absolute value of the relative error. That did not happen until the sample size passed $n^{0.97}$, $n^{0.9995}$ and $n^{0.9995}$ respectively for Near-uniform 40-60, Skewed 20-80 and Skewed 10-90. At those sample sizes, SJPC was always faster in our experiments with any dataset larger than one million records we tried. 
Figure~\ref{fig:runtime-err-vsize} (left) shows the mean running time of SJPC over 10 runs, on Skewed 20-80 and YFCC, as the dataset size is varied. First, for each method, the running time on YFCC closely matched that of Skewed 20-80, and as a result they are not distinguishable in the figure.
Second, as expected, the running time of SJPC grows linearly with the input, whereas the running time for random sampling~\footnote{The sample size was set to $n^{998}$ to have an error not that far from that of SJPC.} increases quadratically with the input.
Each run of random sampling on 8 million records was taking more than 4 days and we could not run it for larger datasets. 
The absolute value of the relative error for both methods are also shown in Figure~\ref{fig:runtime-err-vsize} (right). In terms of the space usage, random sampling requires at least an order of magnitude more space than SJPC, and the space usage of random sampling must increase with data skewness for its error rate to keep up with that of SJPC.


\begin{figure*}[htb]
\centering
 \begin{tabular}{cc}
    \includegraphics[width=3.1in]{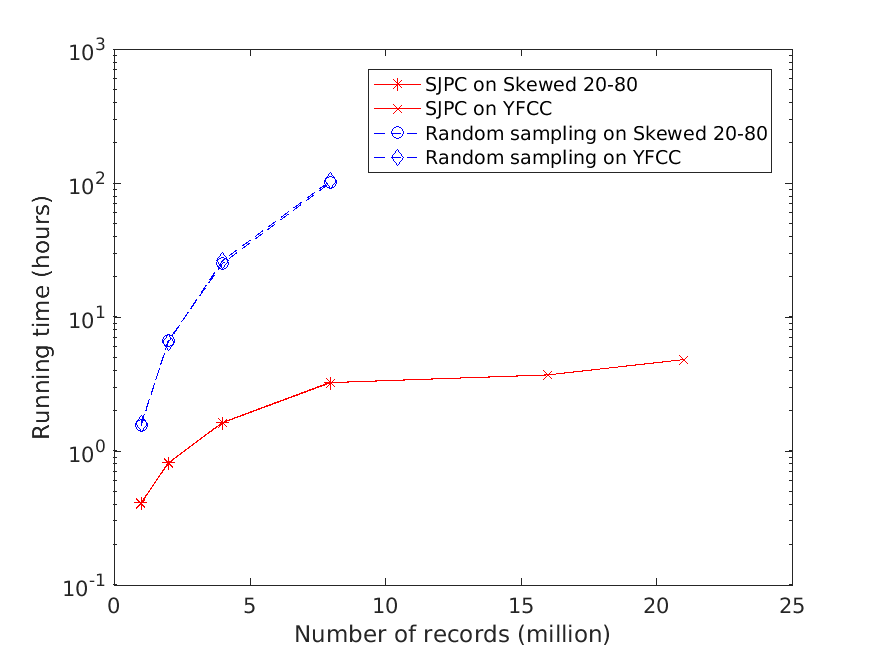} &
    \includegraphics[width=3.1in]{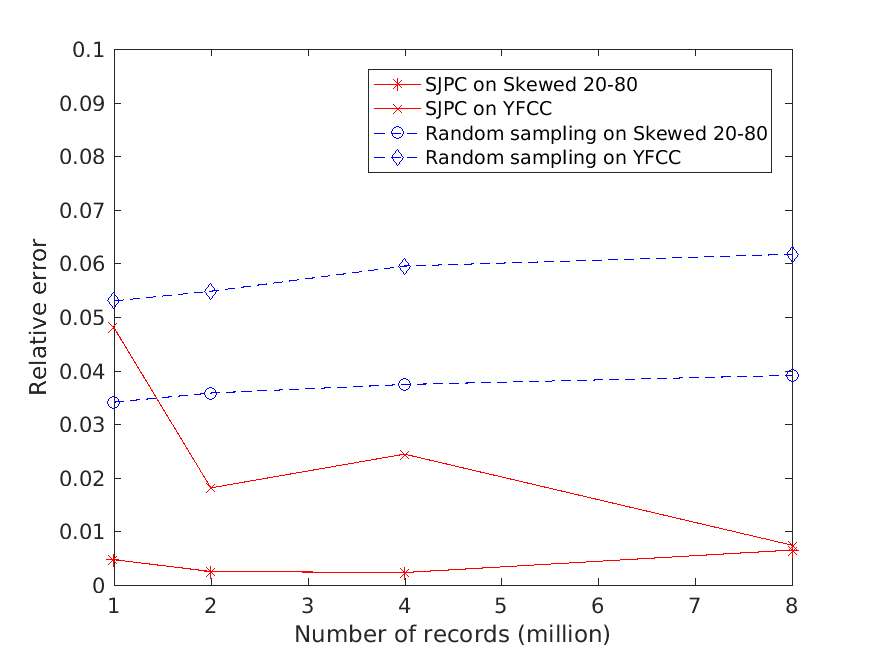}
  \end{tabular}
  \caption{Running time (left) and relative error (right) on Skewed 20-80 and YFCC}
  \label{fig:runtime-err-vsize} 
\end{figure*}

%


%


\subsection{Discussions}
The objective of our experimental evaluation was stated as verifying our analytical findings in more practical settings, and assessing both the robustness and the performance of the SJPC algorithm, as compared to competitors (when applicable).
We evaluated our method on four real datasets, including DBLP5, DBLP6, DBLPtitles and YFCC and some synthetic data including Near-uniform and Skewed, showing that our algorithm outperforms the state-of-the-art methods from the literature  (i.e. LSH-based bucketing in the offline case and random sampling in the online case) in terms of the accuracy of the estimates, the space usage and running time. We experimented with different parameter settings of SJPC, showing that the algorithm is robust and the performance can be managed with different parameters. Our evaluation also confirmed that the SJPC algorithm scales linearly with the input, making it suitable in settings where only one pass over data is feasible. A limitation of the SJPC algorithm is that it does not scale so well to large data dimensionality and this is the price paid for the linear scaling with $n$, for example compared to random sampling. 
  

\section{Related Work}\label{sec:RelatedWork}
Our work relates to the areas of {\it efficient similarity join}, {\it selectivity estimation}, and {\it sketching techniques} .

\noindent
\textbf{Efficient similarity join.}
The problem of similarity join (of tuples) under hamming distance can be mapped to a set similarity join where each tuple becomes a set, for which many algorithms have been developed (e.g., \cite{BayardoMS07,Xiao2011}).
A general and often efficient algorithm to evaluate set similarity join is index nested loop join, where the inner index returns a set of candidates and the outer loop filters those candidates before performing a pairwise comparison to produce the result. For example, all algorithms recently evaluated by Mann et al.~\cite{mann2016empirical} follow this framework and vary in their filtering and candidate generation steps. The time complexity of all these algorithms is quadratic in the input size. To reduce the cost, parallel set similarity is studied using MapReduce~\cite{vernica2010efficient} and with data represented as arrays~\cite{zhao2016similarity}.

Similarity join is also studied in the context of $d$-dimensional points with $d\in [2,32]$. A common approach is to associate points to cubes or cells and only join points with overlapping cells~\cite{koudas2000high,jacox2007spatial}. EGO-based approaches use a combination of sort and divide operations to identify sets of points that cannot join~\cite{kalashnikov2013super}. These algorithms are quadratic for typical values of dimensions and similarity thresholds.

\noindent
\textbf{Selectivity estimation.}
Selectivity estimation has been an important component of query optimization,
and accurate estimates often provide huge savings in cost. Although the problem is widely studied for relational operators with exact predicates (e.g., range predicates~\cite{poosala96}, substring queries~\cite{ChenKKM03jcss}, spatio-temporal queries~\cite{chois02}, joins~\cite{GetoorTK01sigmod}) and despite its importance for similarity predicates (e.g. \cite{Silva09}),
there has not been much study on estimating the selectivity of similarity predicates.
Tata and Patel~\cite{TataP07a} study the problem in the context of Cosine predicates, discussing some of the difficulties.
Hadjieleftheriou et al.~\cite{HYKS08} study the selectivity estimation for set similarity queries
and show that more concise samples can be constructed from the inverted lists of tokens and also report on the performance of different sampling strategies.
Lee et al.~\cite{LNS09,LNS11} study the same problem as ours, and Heise et al.~\cite{Heise14} use random sampling to estimate the sizes of clusters formed by fuzzy duplicates. We compare our work to both that of Lee et al. and random sampling, when applicable or appropriate.

\noindent
\textbf{Sketching techniques.}  
As our work uses sketching to estimate the size of a sub-value stream, there are quite some works on sketching techniques that are applicable. For example, instead of Fast-AGMS~\cite{CountSketch04TCS,SketchNet05VLDB}, which is used in our experiments, Bloom filters may be extended to answer frequency related queries including join and self-join size estimation~\cite{CountMin05JAlgorithms,cmm}.
Rusu and Dobra~\cite{RusuD08} review and evaluate some of these sketches
for join size estimation.
The same authors also study the problem of sketching over samples and show that a speed up in factors of 10 is achievable without much decrease in accuracy~\cite{RusuD09}.
Our sampling is slightly different in that we are sampling from the space of projections of each record.

\noindent
\textbf{Others.}
Deng et al.~\cite{fanDiversity12} study the problem of diversity analysis where similar randomized techniques are used to estimate an average pair-wise similarity. String similarity join~\cite{jiang2014string} may also be mapped to set similarity (for token-based) or hamming similarity (for character-based), where join size estimation techniques will be useful.
Our work may also be applicable in data cleaning and record deduplication settings (see \cite{ElmagarmidIV07tkde} and Christen~\cite{christen2012survey} for extensive surveys). 
\section{Conclusions and Future Directions}\label{sec:conclude}
In this paper, we studied the problem of similarity self-join size estimation and presented 
a solution for efficiently finding an estimate within one pass over data. We analyzed the accuracy, time and space usage of our algorithm and experimentally evaluated it on both real and synthetic datasets. Our evaluation showed that the proposed algorithm has a relatively high accuracy (often an order of magnitude better than the competitors) and low time and space cost.
 
Our algorithm scales linearly with the input, and even larger input sizes can help with the accuracy, but it does not scale so well with the dimensionality, which is the price paid for the linear scale up with $n$. 
Our method is readily applicable in cases where $d$, the dimensionality of the data (or the number of columns), is low, or the similarity threshold $s$ is high so that ${d \choose s}$ does not explode to avoid {\it the curse of dimensionality}. 
On the other hand, when the input has a large number of columns, it is often the case that a subset of the columns are selected in queries or analyzed (this has been the premise in some of the work on
{\it projected clustering}~\cite{AggarwalHWY05} and detecting unique column combinations~\cite{AbedjanQN14}). 


More studies are needed to understand the behaviour of our algorithm, applied to high dimensional data, and the conditions under which more accurate estimates can be obtained. 
In particular, one area is studying some of the conditions under which our work can be extended to higher dimensions. For example, one may decompose a table into smaller attribute groupings, and compute the similarity self-join size under each grouping before merging the results. Finding decompositions under which the similarity self-join size can be accurately estimated from that of the decomposed table is an interesting future direction. Another area is studying the problem and the proposed solution under some simplifying assumptions (e.g. on the data distribution) that allows tighter bounds to be obtained and/or a better understanding of the problem is gained. One more interesting question is if (the data structure or the estimate of) a similarity join size estimation can be part of a similarity join algorithm, possibly to speed up the join. 


\section*{Acknowledgments}
This research is supported by the Natural Sciences and Engineering
Research Council of Canada.

\bibliographystyle{IEEEtran}
\bibliography{ref}

\ifdefined\Archive
\else
\begin{wrapfigure}{L}{0.10\textwidth}
\centering
\includegraphics[width=0.10\textwidth]{davood.jpg}
\end{wrapfigure}
\noindent  {\bf Davood Rafiei} 
did his undergrad work at Sharif, his M.Sc. in Waterloo and his PhD in Toronto before joining the University of Alberta, where he is now Associate Professor of Computer Science and member of the Database Systems Research Group. 
His areas of interest span over databases and the Web and include 
integrating natural language text with relational data, Web information 
retrieval and similarity queries and indexing. Davood has spent time, as a visiting scientist, 
at Google (Mountain View), Kyoto University and the University of Paris Descartes. \\

\begin{wrapfigure}{L}{0.10\textwidth}
\centering
\includegraphics[width=0.10\textwidth]{fan.jpeg}
\end{wrapfigure}
\noindent {\bf Fan Deng} did both his undergrad and his M.Sc. at Huazhong University of Science and Technology and his PhD at the University of Alberta. He was a postdoc fellow at L3S before moving to industry. His areas of interest include database systems, information retrieval and social networks.
\fi

\ifdefined\Archive

\newpage
\appendix
\section{Proofs}
\label{sec:proofs}
\noindent
\textit{Lemma~\ref{lem:sampling_lower_bound}}
Random-sampling requires a sample of size $\Omega(\sqrt{n})$ to give an 
estimate of the similarity self-join size with a relative error less
than $100\%$ with high probability. 
\begin{proof}
This is an adaptation of the proof of Lemma 2.3 in \cite{JoinSize99PODS}.
Let $s$ be an arbitrary similarity threshold.
Construct two datasets $D_s$ and $D_{ns}$, each with $n$ records such that no record in $D_{ns}$ is $k$-similar to any other record for all values of $k$, but $D_s$ has $n/2$ $s$-similar pairs of records and there is no other form of similarity between the records. 
A sampling-based estimate of the similarity self-join size for $D_{ns}$ will be $n$, and that for $D_s$ will be also $n$ using samples of size $o(\sqrt{n})$ with high probability. This is simply because the 
chance that a similar pair (not including self-pairs) makes to the sample is $(n/2)/(n(n-1)/2)$ and the expected number of such pairs in a sample of size $o(\sqrt{n})$ is $O(n^{3/2}/n^2)$, which is close to zero for large $n$.  
However, the similarity self-join sizes for $D_{ns}$ and $D_s$ are $n$ and $2n$ respectively, and the estimate is off by a factor of $2$ with high probability.
As another instance, suppose $D_s$ has $\sqrt{n}$ records that are identical on $s$ columns and $D_{ns}$ is as defined before. The s-similarity self-join sizes of $D_{ns}$ is $n$ and that of $D_s$ is $2n$. However,
the chance that one of those s-similar pairs is included in a sample of size $\sqrt{n}$ is $\sqrt{n}/n$, 
and the chance that $k$ of them are included in a sample of size $\sqrt{n}$ is
\[
\frac{\sqrt{n}(\sqrt{n}-1)\ldots(\sqrt{n}-k+1)}{n(n-1)\ldots (n-k+1)}.
\]
This probability is very close to zero for large values of $n$ or $k$.
That means random sampling will report with a high probability an s-similarity self-join size of $n$ for both $D_s$ and $D_{ns}$. 
\end{proof}

\noindent
\textit{Theorem \ref{thm:E_VAR_offline}. }
The SJPC algorithm gives an unbiased estimate of the 
$s$-similarity self-join size under the offline scenario, i.e. $E[G_s] = g_s$, and
the standard deviation of $\frac{G_s} {g_s}$ is at most 
$${d \choose s}\sqrt{\frac{1}{r} {2(d-s) \choose {d-s}} / g_s}, $$
where $G_s$ is the estimate and $g_s$ is the true value.
\begin{proof}
To find the variance of $G_s$, we need the variance of 
$X_k$ ($k=s \dots d$). Eq.~\ref{eq:X_k} gives a recursive
expression of $X_k$, as a function of $X_{k+1} \dots X_d$ and $Y_k$.
First, we show that $X_k$ can be represented as a function 
of $Y_k \dots Y_d$ with the recursion removed. Second, we prove the unbiased
property of $X_k$. Last, we derive an upper bound of the variance of $Y_k$, and this allows us
to bound the variance of $G_s$.
The details are as follows.

First, we prove by induction that
\begin{equation} \label{eq:X_k_Y} 
X_k = \frac{1}{r^2} \sum_{j=k}^d (-1)^{j-k} {j \choose k} Y_j + C_k,
\end{equation}
where $k \in [1,d]$, and $C_k$ is a constant hence not important in 
the expression of the variance.
From Eq.~\ref{eq:X_k}, we can easily verify that Eq.~\ref{eq:X_k_Y} holds 
for $k=d$ and $d-1$. Assuming Eq.~\ref{eq:X_k_Y} holds for an arbitrary $k \in [2, d]$,
we want to prove that it holds for $k-1$ as well.
From Eq.~\ref{eq:X_k} we have 
\[
X_{k-1} = (Y_{k-1} - r {d \choose k-1} n) / r^2 - \sum_{j=k}^d {j \choose k-1} X_j
\]
Using the induction hypothesis to replace $X_j$, we have
\begin{align*}
X_{k-1}=& \frac{1}{r^2} Y_{k-1} - \frac{n}{r} {d \choose k-1} \\ 
 &  - \sum_{j=k}^d {j \choose k-1} 
   (\frac{1}{r^2} \sum_{i=j}^d (-1)^{i-j} {i \choose j} Y_i + C_j).
\end{align*}
If we change the indexes to the filled part in Figure~\ref{fig:indexes}(a) and denote the constants with
$C_{k-1}$, the right side becomes
\begin{align*}
\frac{1}{r^2} Y_{k-1} - \frac{1}{r^2} \sum_{i=k}^d \sum_{j=k}^i
   (-1)^{i-j} {j \choose k-1} {i \choose j} Y_i + C_{k-1}.
\end{align*}
It is easy to verify that ${j \choose k-1}{i \choose j}={i \choose k-1}{i-k+1 \choose j-k+1}$. 
Also\\ $\sum_{j=k}^i (-1)^{i-j} {i-k+1 \choose j-k+1} = (-1)^{i-k}$ for $i=k,\ldots,d$ (see Lemma~\ref{lemma:binomialT2} in the Appendix), hence
\begin{align*}
X_{k-1}=& \frac{1}{r^2} Y_{k-1} - \frac{1}{r^2} \sum_{i=k}^d {i \choose k-1} Y_i
   (-1)^{i-k} + C_{k-1}\\
=& \frac{1}{r^2} \sum_{i=k-1}^d (-1)^{i-k+1} {i \choose k-1}Y_i  + C_{k-1}.
\end{align*}
Eq.~\ref{eq:X_k_Y} holds for $k-1$, thus
it holds for all $k \in [1,d]$. Now the similarity self-join size, $G_s$, can be rewritten as follows (with the replacement of indexes to the filled part of Figure~\ref{fig:indexes}(b) in the last step):
\begin{align*} 
G_s = \sum_{k=s}^d X_k = \frac{1}{r^2} \sum_{k=s}^d 
      \sum_{j=k}^d (-1)^{j-k} {j \choose k} Y_j + \sum_{k=s}^d C_k + n
\end{align*}
\begin{equation}
G_s =  \frac{1}{r^2} \sum_{j=s}^d \sum_{k=s}^j (-1)^{j-k} {j \choose k} Y_j  + \sum_{k=s}^d C_k + n.
\label{eq:Gs}
\end{equation}
%
%
\begin{figure}[ht]
\centering
\includegraphics[width=1.8in]{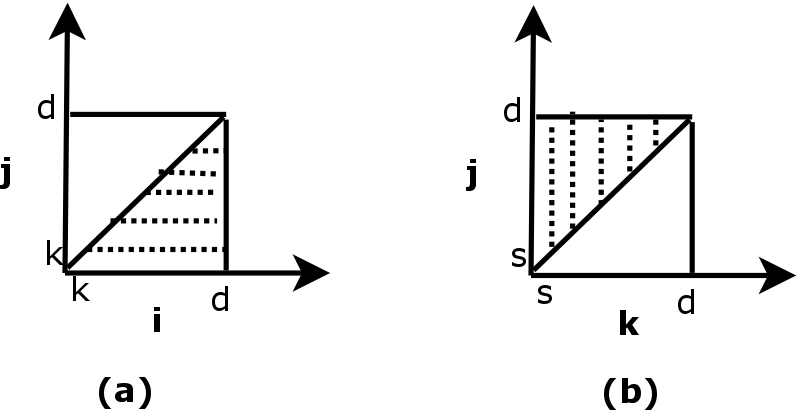}
\caption{Index substitutions}
\label{fig:indexes}
\end{figure}

Second, we show that SelfJoinPairCount gives an unbiased estimate.
Let $O_j$ be the set of all $j$-similar record pairs excluding self pairs (i.e. when a record joins itself),
and $Z_{o_j, k}$ be the value that a $j$-similar record pair, denoted by 
$o_j$, 
contributes to  $Y_k$ (the self-join size of the $k$-sub-value stream);
then we have
\begin{equation}
Y_k = \sum_{j=k}^d \sum_{\forall o_j \in O_j} Z_{o_j, k} + n r {d \choose k}.
\end{equation}
Note that $X_k$, $Y_k$ and $Z_{o_j, k}$ are all random variables. 
The expected value of $Z_{o_j, k}$ in the sample is
$$\mu_{j,k} = E[Z_{o_j, k}] = 
r^2 {j \choose k}. 
$$
Therefore
\begin{equation} 
E[Y_k] = r^2 \sum_{j=k}^d {j \choose k} x_j + n r {d \choose k}.
\label{eq:EYk}
\end{equation}
$\mu_{j,k}$ is the expected value that a $j$-similar pair contributes
to $Y_k$ and $x_j$ is the true number of $j$-similar record pairs. From Eq.~\ref{eq:X_k} we can see that SelfJoinPairCount removes the 
contributions of \{$k+1, k+2, \dots, d$\}-similar pairs from $Y_k$, thus
it is not hard to verify that $X_k$ is an unbiased estimate for $x_k$.

Last, we derive an upper bound of the variance of $Y_k$.
Let $l_k$ denote the expected number of times a record will appear in a sample of level
$k$, i.e. $l_k= r {d \choose k}$, and
$p_{j, k, i}$ be the probability that a $j$-similar record pair contributes
$i$ to  $Y_k$, then the variance of $Z_{o_j, k}$ is
\begin{align*}
\sigma_{j,k}^2 = & \mbox{VAR}[Z_{o_j, k}] = E[Z_{o_j, k}^2] - \mu_{j,k}^2 = 
            \sum_{i=1}^{l_k} i^2 p_{j, k, i} - \mu_{j,k}^2 \\
      \leq &  l_k \sum_{i=1}^{l_k} i p_{j, k, i} - \mu_{j,k}^2 
	    = l_k \mu_{j,k} - \mu_{j,k}^2 \\
	 = & r^3 {j \choose k} {d \choose k}- r^4 {j \choose k}^2 
	    \leq r^3 {j \choose k} {d \choose k}.
\end{align*}
The variance of $Y_k$ can be written as
\[
\mbox{VAR}[Y_k] = \sum_{j=k}^d x_j \sigma_{j,k}^2  + 2\sum_{\substack{j1=k,\ldots,d\\ j2=k,\ldots,d\\ o,o'\in O \ \&\ o \neq o'}} Cov(Z_{o_{j1},k}, Z_{o'_{j2},k}).
\]
%
%
For two pairs $o$ and $o^\prime$, they may or may not have a row in common. When the two pairs have no row in common, the covariance term will be zero. Now suppose there is a row $r$ that is common between the two pairs. Let's denote the pairs as $(r, r_1)$ and $(r, r_2)$.
Since the projections of $r_1$ and $r_2$ are independently chosen uniformly at random, the covariance term is zero even though the projections of $r$ is the same for both pairs. Hence the covariance term can be ignored, and we have
$$
\mbox{VAR}[Y_k] = \sum_{j=k}^d x_j \sigma_{j,k}^2 
           \leq r^3 {d \choose k} \sum_{j=k}^d {j \choose k} x_j. 
$$
Using Equation~\ref{eq:Gs}, we can bound $\mbox{VAR}[G_s]$ from above to
\begin{align*} 
      \frac{1}{r^4} \sum_{k=s}^d {k \choose s}^2 \mbox{VAR}[Y_k] 
      \leq & \frac{1}{r} \sum_{k=s}^d  {k \choose s}^2 
	      \sum_{j=k}^d x_j {j \choose k}{d \choose k} 
\end{align*}
\begin{align*}
         = & \frac{1}{r} \sum_{k=s}^d \sum_{j=k}^d x_j {j \choose s}
	     {j-s \choose k-s}{d \choose s} {d-s \choose k-s} 
\end{align*}
\begin{align*}
	 \leq & \frac{1}{r} {d \choose s}^2 \sum_{k=s}^d {d-s \choose k-s}^2
	        \sum_{j=k}^d x_j
		\leq \frac{1}{r} {d \choose s}^2 \sum_{j=s}^d x_j 
		\sum_{k=s}^d {d-s \choose k-s}^2 \\
	 = &\frac{1}{r} {d \choose s}^2 {2(d-s) \choose d-s} g_s.
\end{align*}
Thus 
\[
\mbox{VAR}(G_s/g_s) =\frac{1}{r} {d \choose s}^2 {2(d-s) \choose d-s} / g_s.
\]
\end{proof}

\begin{lemma}
\label{lemma:binomialT2}
For $i \geq k$, 
\[
\sum_{j=k}^i (-1)^{i-j} {i-k+1 \choose j-k+1} = (-1)^{i-k}.
\]
\end{lemma}
\begin{proof}
If we replace the variables $j-k+1$ with $m$ and $i-k+1$ with $n$, the left side becomes
\begin{align*}
 & \sum_{m=1}^n (-1)^{n-m} {n \choose m} \\
 = & -(-1)^n +  \sum_{m=0}^n (-1)^{n-m} {n \choose m}\\
 = & (-1)^{n-1} + (-1)^n  \sum_{m=0}^n (-1)^m {n \choose m}.
\end{align*}
With the Binomial theorem~\cite{binothm} applied to the summation, the second term becomes zero. The proof is complete after replacing $n-1$ in the first term with $i-k$.
\end{proof}
\noindent
\textit{Theorem}~\ref{thm:online-SJPC}
(Unbiased estimate and variance)
The SJPC algorithm gives an unbiased estimate of the
$s$-similarity self-join size in an online scenario, i.e. $E[G_s] = g_s$, and the variance
of $\frac{G_s} {g_s}$ is at most
$$
{d \choose s}^2 \frac{1}{r}  {2(d-s) \choose d-s} 
((1+\frac{2}{w})/ g_s + \frac{2}{w} (1 + \frac{n}{r g_s})^2),
$$
where $w$ is the Fast-AGMS sketch width (depth is $1$), 
$d$ is the number of attributes, $s$ is the given similarity 
threshold, $r$ is the sampling ratio, $g_s$ is the true value
of the similarity self-join size, and $G_s$ is the estimated value.
\begin{proof}
Since both offline version of SJPC and Fast-AGMS provides unbiased estimates,
it is not hard to see the estimates from the online case are also unbiased.
 
Let $Y_k^\prime$ denote the self-join size estimate of 
the sub-value stream 
using the Fast-AGMS algorithm, and $Y_k$ denote  
the self-join size estimate using the offline version of our algorithm
as before. According to the law of 	total variance~\cite{weiss05}, 
\begin{align*}
\mbox{VAR}[Y_k^\prime] =& E[\mbox{VAR}[Y_k^\prime | Y_k]] + \mbox{VAR}[E[Y_k^\prime | Y_k]] \\
            \leq&  E[\frac{2}{w} Y_k^2] + \mbox{VAR}[Y_k] \\
		=& (1+\frac{2}{w}) \mbox{VAR}[Y_k] + \frac{2}{w} E[Y_k]^2.
\end{align*}
Similar to the proof of Theorem~\ref{thm:E_VAR_offline}, we have
\begin{align*} 
         & \mbox{VAR}[G_s] = \mbox{VAR}[\sum_{k=s}^d X_k] 
 	 \leq \frac{1}{r^4} \sum_{k=s}^d {k \choose s}^2 \mbox{VAR}[Y^\prime_k] \\
   \leq & \frac{1}{r^4} (1+\frac{2}{w}) \sum_{k=s}^d {k \choose s}^2 \mbox{VAR}[Y_k]
         + \frac{1}{r^4} \frac{2}{w}  \sum_{k=s}^d {k \choose s}^2 E[Y_k]^2 
\end{align*}
With 
$\frac{1}{r^4} \sum_{k=s}^d {k \choose s}^2 \mbox{VAR}[Y_k] \leq \frac{1}{r} {d \choose s}^2 {2(d-s) \choose d-s}g_s$
as shown above and replacing $E[Y_k]$ from Eq.~\ref{eq:EYk}, we have
\begin{align*}
  \mbox{VAR}[G_s] \leq & (1+\frac{2}{w}) \frac{1}{r} {d \choose s}^2 {2(d-s) \choose d-s}g_s\\
        &  + \frac{1}{r^4} \frac{2}{w} \sum_{k=s}^d {k \choose s}^2 
	  (r^2  \sum_{j=k}^d {j \choose k} x_j + n r {d \choose k})^2 \\
   \leq & (1+\frac{2}{w}) \frac{1}{r} {d \choose s}^2 {2(d-s) \choose d-s}g_s\\
        &   + \frac{1}{r^4} \frac{2}{w} \sum_{k=s}^d {k \choose s}^2 
	   {d \choose k}^2 (r^2 g_s + nr)^2 \\
      = & (1+\frac{2}{w}) \frac{1}{r} {d \choose s}^2 {2(d-s) \choose d-s}g_s\\
        & + \frac{1}{r^2 } \frac{2}{w} \sum_{k=s}^d {d \choose s}^2 
	  {d-s \choose k-s}^2 (r g_s + n)^2 \\
      = & (1+\frac{2}{w}) \frac{1}{r} {d \choose s}^2 {2(d-s) \choose d-s}g_s\\
        &  + \frac{1}{r^2 } \frac{2}{w} {d \choose s}^2 {2(d-s) \choose d-s}
	  (r g_s + n)^2 .
\end{align*}
Therefore, the claim on the variance of $\frac{G_s} {g_s}$ holds.
\end{proof}


\noindent
\textit{Theorem}~\ref{thm:selectivity_online}
(Space and time cost to bound the selectivity estimation error)
The SJPC algorithm guarantees that the
estimated selectivity of the similarity self-join deviates from the true value
by at most $\epsilon$ with probability at least $1-\lambda$. More precisely,
$Pr[|\hat{\theta}_s - \theta_s| \leq \epsilon] \ge 1-\lambda$, where
$\hat{\theta}_s$ is the estimated selectivity and $\theta_s$ is the true
value. The space cost is $O(\log(1/\lambda) (d-s+1) w)$,
and the time cost for processing each record is
$O(\log(1/\lambda) {d \choose s}^2 {2(d-s) \choose d-s} 
(\sum_{k=s}^d {d \choose k}) / (\epsilon^2 w)) $.
\begin{proof} 
We have
\begin{align*}
& Pr[|G_s - g_s| > \epsilon n^2]  \leq  \frac{\mbox{VAR}[G_s]}{\epsilon^2 n^4} \\
\leq  & \frac{1}{\epsilon^2 n^2 r} {d \choose s}^2 {2(d-s) \choose d-s}
        ((1+\frac{2}{w}) + \frac{2}{w} (\frac{g_s}{n} + \frac{1}{r} )^2 )\\
\leq  & \frac{1}{\epsilon^2 r} {d \choose s}^2 {2(d-s) \choose d-s}
        (\frac{1}{n^2} (1+ \frac{2}{w}) + \frac{2}{w} (1+\frac{1}{nr})^2).
\end{align*}
Let the above term be $\frac{1}{8}$,  then we have
\begin{align*}
r \ge & \frac{8}{\epsilon^2} {d \choose s}^2 {2(d-s) \choose d-s}
   (\frac{2}{n^2} (1+ \frac{2}{w}) + \frac{8}{w}) \\
\approx & \frac{64}{\epsilon^2 w} {d \choose s}^2 {2(d-s) \choose d-s},
\end{align*}
assuming $n r \ge 1$. 
To increase the success probability, 
we can repeat the same algorithm independently
$2 \log(1/ \lambda)$ times, 
and take the median of the multiple results. Due to Chernoff bound, 
we can guarantee the probability that SJPC fails is at most $\lambda$.
Since SJPC picks $O(r \sum_{k=s}^d {d \choose k})$ 
sub-values for each record,
the time and space costs stated in the theorem can give the desired result. 
\end{proof}

\fi  

\end{document}